\documentclass[onecolumn, a4size, 11pt]{IEEEtran}
\usepackage{amsmath}
\usepackage{amssymb}
\usepackage{amsfonts}
\usepackage{graphicx}
\usepackage{epsfig}
\usepackage{subfigure}
\usepackage{psfrag}
\usepackage{xcolor}

\title{Optimized Uplink Transmission in Multi-Antenna C-RAN with Spatial Compression and Forward
\footnote {L. Liu is with the Department of Electrical
and Computer Engineering, National University of Singapore
(e-mail:liu\_liang@nus.edu.sg).}\footnote{R. Zhang is with the
Department of Electrical and Computer Engineering, National
University of Singapore (e-mail:elezhang@nus.edu.sg). He is also
with the Institute for Infocomm Research, A*STAR, Singapore.}}

\author{Liang Liu and Rui Zhang}

\begin{document}
\maketitle \thispagestyle{empty} \vspace{-0.3in}

\begin{abstract}
Massive multiple-input multiple-output (MIMO) and cloud radio access network (C-RAN) are two promising techniques for implementing future wireless communication systems, where a large number of antennas are deployed either being co-located at the base station (BS) or totally distributed at separate sites called remote radio heads (RRHs) to achieve enormous spectrum efficiency and energy efficiency gains. In this paper, we consider a general antenna deployment design for wireless networks, termed multi-antenna C-RAN, where a flexible number of antennas can be equipped at each RRH to more effectively balance the performance and fronthaul complexity trade-off beyond the conventional massive MIMO and single-antenna C-RAN. In order to coordinate and control the fronthaul traffic over multi-antenna RRHs, under the uplink communication setup, we propose a new ``spatial-compression-and-forward (SCF)'' scheme, where each RRH first performs a linear spatial filtering to denoise and maximally compress its received signals from multiple users to a reduced number of dimensions, then conducts uniform scalar quantization over each of the resulting dimensions in parallel, and finally sends the total quantized bits to the baseband unit (BBU) via a finite-rate fronthaul link for joint information decoding. Under this scheme, we maximize the minimum signal-to-interference-plus-noise ratio (SINR) of all users at the BBU by a joint resource allocation over the wireless transmission and fronthaul links. Specifically, each RRH determines its own spatial filtering solution in a distributed manner to reduce the signalling overhead with the BBU, while the BBU jointly optimizes the users' transmit power, the RRHs' fronthaul bits allocation, and the BBU's receive beamforming with fixed spatial filters at individual RRHs. Through numerical results, it is shown that given a total number of antennas to be deployed, multi-antenna C-RAN with the proposed SCF and joint optimization significantly outperforms both massive MIMO and single-antenna C-RAN under practical fronthaul capacity constraints.
\end{abstract}


\setlength{\baselineskip}{1.3\baselineskip}
\newtheorem{definition}{\underline{Definition}}[section]
\newtheorem{fact}{Fact}
\newtheorem{assumption}{Assumption}
\newtheorem{theorem}{\underline{Theorem}}[section]
\newtheorem{lemma}{\underline{Lemma}}[section]
\newtheorem{corollary}{\underline{Corollary}}[section]
\newtheorem{proposition}{\underline{Proposition}}[section]
\newtheorem{example}{\underline{Example}}[section]
\newtheorem{remark}{\underline{Remark}}[section]
\newtheorem{algorithm}{\underline{Algorithm}}[section]
\newcommand{\mv}[1]{\mbox{\boldmath{$ #1 $}}}

\section{Introduction}\label{sec:Introduction}

\begin{figure}
\begin{center}
\subfigure[Massive MIMO]
{\scalebox{0.5}{\includegraphics*{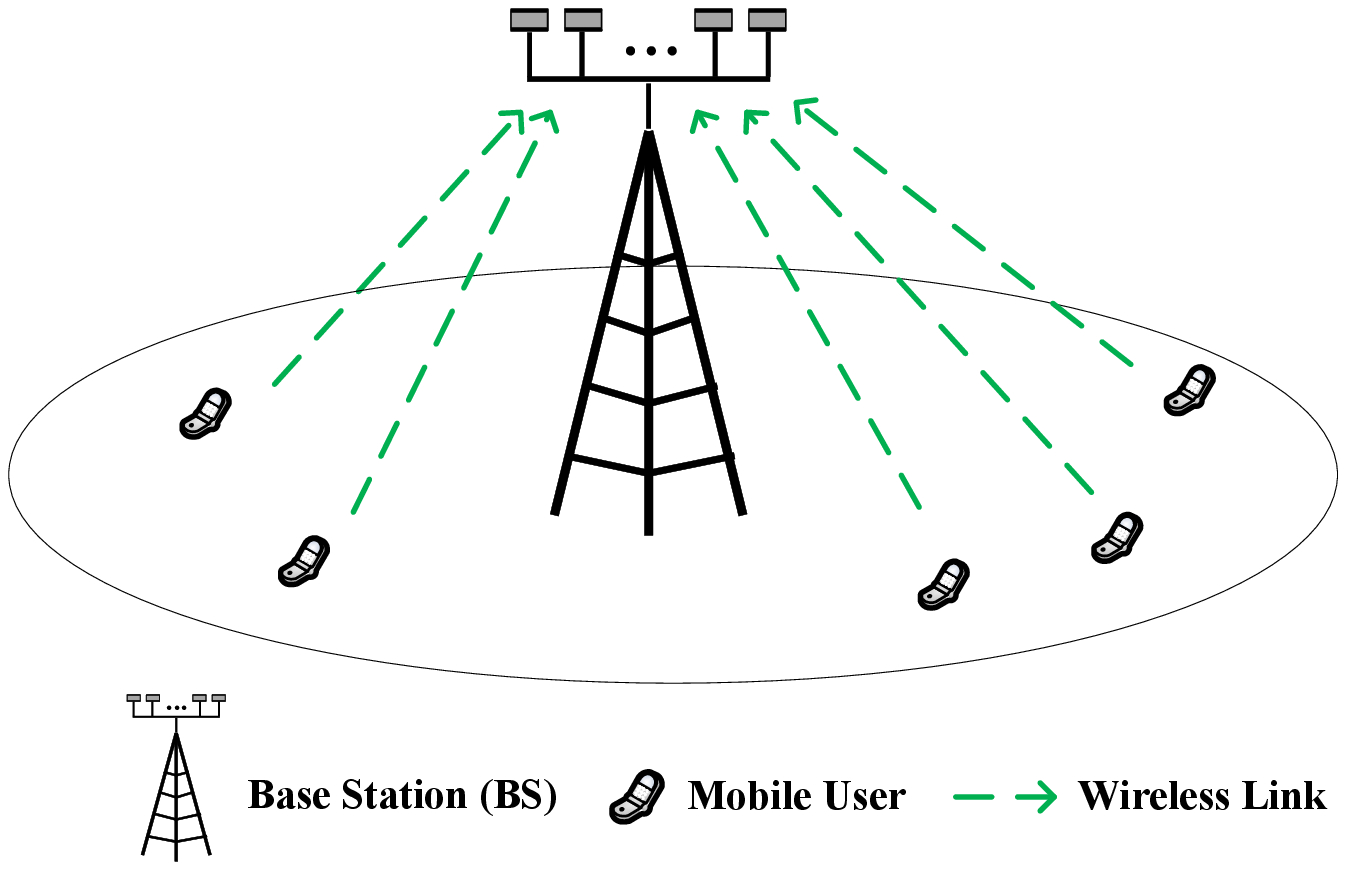}}} \\
\subfigure[Single-antenna C-RAN]
{\scalebox{0.5}{\includegraphics*{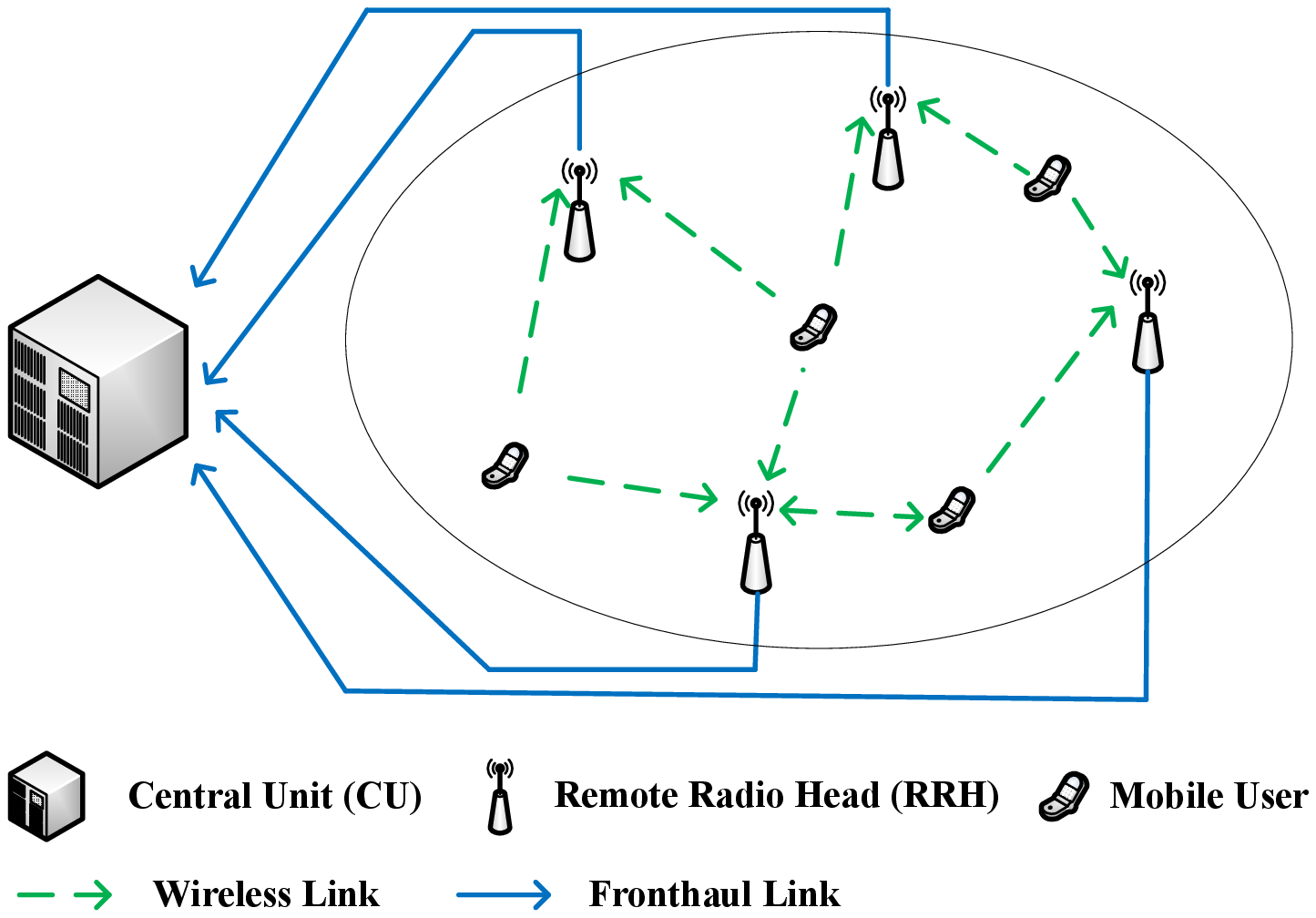}}} \ \ \ \ \
\subfigure[Multi-antenna C-RAN]
{\scalebox{0.5}{\includegraphics*{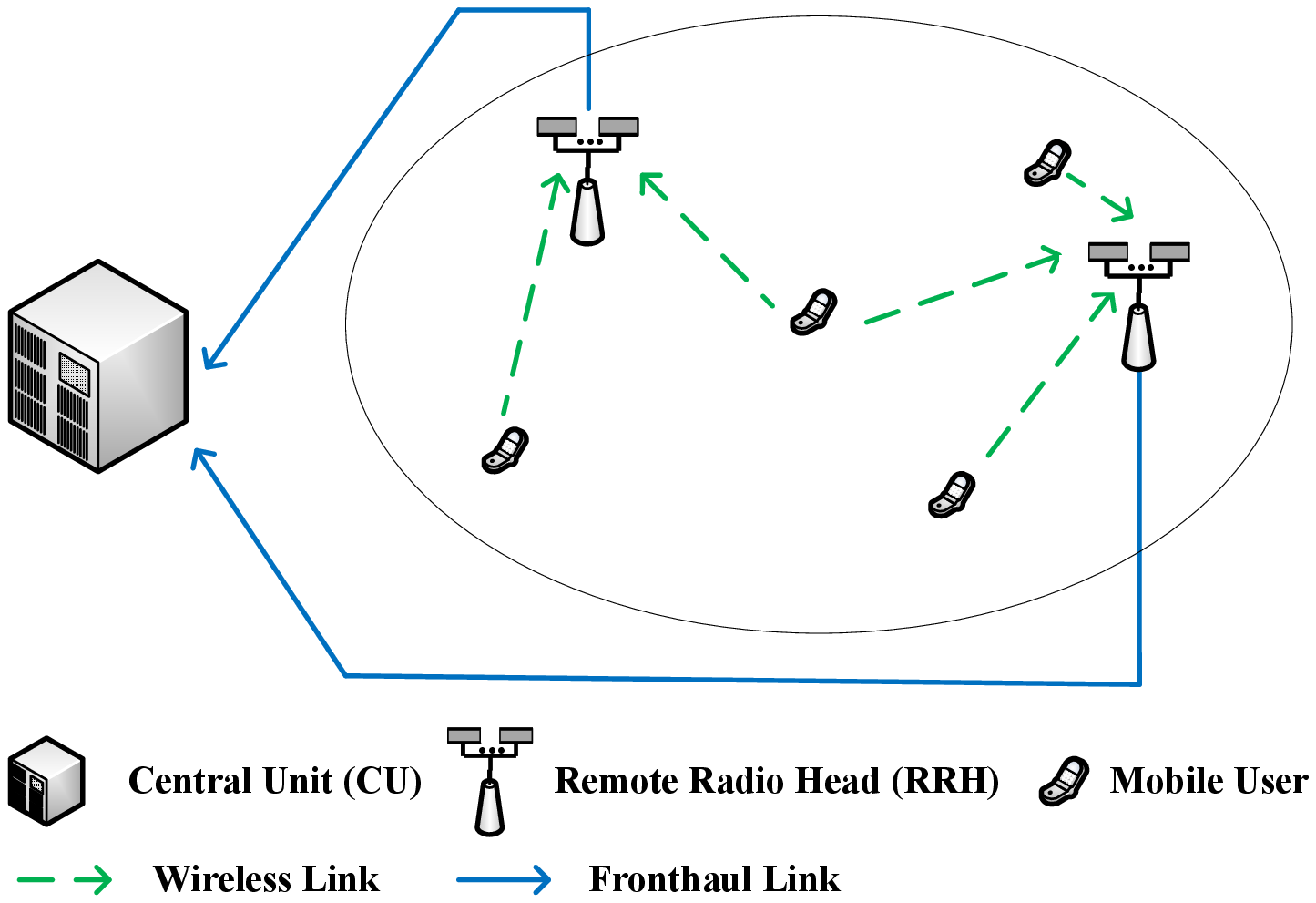}}}
\end{center}
\caption{Comparison of massive MIMO, single-antenna C-RAN, and multi-antenna C-RAN in the uplink transmission.}\label{fig8}
\end{figure}

It is anticipated that the mobile data traffic will grow 1000 times higher in the next decade with a rate of roughly a factor of two per year \cite{3GPP}. As a result, the fifth-generation or 5G wireless system on the roadmap has to make a major paradigm shift to accommodate this dramatic explosion in future demands for wireless data communications. Both as highly promising candidate techniques for the 5G wireless network, massive multiple-input multiple-output (MIMO) \cite{Marzetta10,Rui14} and cloud radio access network (C-RAN) \cite{ChinaMobile} have recently drawn significant attentions (see Fig. \ref{fig8}). Both techniques advocate the use of very large number of antennas and centralized signal processing to achieve enormously improved efficiency in spectrum and energy usage. However, they are practically implemented with different performance and complexity trade-off considerations.

In a massive MIMO system, as depicted in Fig. \ref{fig8}(a), a large number of antennas are deployed at the base station (BS) to reap all the benefits of the conventional multi-user MIMO system but with a much greater scale. It is shown in \cite{Marzetta10} that in the asymptotic regime where the number of antennas at the BS is much larger than that of the users, the channels of different users are orthogonal with each other, and thus a simple matched filter, i.e., maximal ratio combining (MRC) in the uplink and maximal ratio transmission (MRT) in the downlink, is capacity-achieving. Moreover, the required transmit energy per bit vanishes as the number of antennas goes to infinity \cite{Erric13}. Despite of the above appealing benefits, there are many issues that potentially limit the performance of massive MIMO. For instance, due to the limited space for installation, the number of antennas at the BS is finite and their channels are generally correlated; as a result, the theoretical performance limit in the asymptotic regime cannot be fully achieved in practice with simple MRC/MRT. Moreover, similar to conventional MIMO system, the performance of the cell-edge users in the massive MIMO system is still a bottleneck since their channel conditions may be weak due to more significant path attenuation.

On the other hand, in a single-antenna C-RAN as shown in Fig. \ref{fig8}(b), each antenna is deployed at one separate site called remote radio head (RRH), which is connected to one central unit (CU) via a high-speed fronthaul link (fiber or wireless) for joint information processing. Unlike the BS in the massive MIMO system which decodes or encodes the user messages locally, in C-RAN each RRH merely forwards the signals to/from the CU via its fronthaul link, while leaving the joint encoding/decoding complexity to a baseband unit (BBU) in the CU. As a result, if the data can be perfectly transmitted between the BBU and RRHs over the fronthaul links without any error, the single-antenna C-RAN can be viewed as another form of massive MIMO but with the antennas densely distributed over the whole served region such that all the mobile users can be served by some adjacent antennas with strong channel conditions. However, one obstacle constraining the practically achievable throughput of C-RAN is the limited-capacity fronthaul link between each RRH and BBU. According to \cite{ChinaMobile}, the capacity of the current commercial fibers, which is about several Gbps (gigabits per second), can be easily overwhelmed in C-RAN even under moderate data traffic; whereas wireless fronthaul techniques such as those exploiting the millimeter-wave bands in general operate with even lower data rate per link and shorter distance as compared to the fiber based solution. In the literature, a considerable amount of effort has been dedicated to study fronthaul quantization/compression techniques in the uplink communication \cite{Steinberg09}-\cite{Yu13}. Specifically, the ``quantize-and-forward (QF)'' scheme \cite{Tse11}-\cite{Kramer08} is adopted to reduce the communication rates between the BBU and RRHs, where each single-antenna RRH samples, quantizes and forwards its received signals to the BBU over its fronthaul link. Moreover, the QF scheme is usually studied under an information-theoretical Gaussian test channel model, and the quantization noise levels of all the RRHs are jointly optimized to maximize the end-to-end throughput subject to the capacity constraints of individual fronthaul links. Recently, \cite{Liu14} reveals that subject to the capacity of the current commercial fibers, the above theoretical performance upper bound can be approached by applying a simple uniform scalar quantization at each RRH for orthogonal frequency division multiplexing (OFDM) based single-antenna C-RAN. Moreover, another line of research in C-RAN studies the control of the overall fronthaul traffic by selectively activating only a small subset of all the RRHs to meet the user demands \cite{RuiZhang}.

In the aforementioned massive MIMO system and single-antenna C-RAN, the antennas are deployed in a totally co-located and separated manner, respectively. In addition to these two extreme cases, a more general solution is to distribute the antennas at a reduced number of sites as compared to single-antenna C-RAN, while the RRH at each site is equipped with multiple antennas. This so-called multi-antenna C-RAN design can provide both improved channel conditions to the users as in single-antenna C-RAN as well as local spatial multiplexing gain as in massive MIMO; furthermore, by adjusting the number of antennas at each RRH and the number of RRHs, these benefits can be flexibly traded-off with a given total number of antennas deployed. However, the fronthaul issue becomes more severe in the multi-antenna C-RAN as compared to conventional single-antenna counterpart, since the limited fronthaul capacity at each RRH needs to be allocated over the outputs of more antennas, thus leading to more quantization errors in general. To tackle this challenge and address the above design trade-off, in this paper we focus our study on the uplink communication in C-RAN with multi-antenna RRHs each subject to a finite fronthaul capacity constraint. The contributions of this paper are summarized as follows.
\begin{itemize}
\item We propose a novel ``spatial-compression-and-forward (SCF)'' scheme for the multi-antenna RRHs to balance between the information conveyed to the BBU and the data traffic over the fronthaul links. Specifically, each RRH first performs spatial filtering to its received signals at different antennas from the users such that the signals are denoised and maximally compressed into a reduced number of dimensions. Then, the RRH applies the simple uniform scalar quantization \cite{Liu14} over each of these dimensions in parallel by appropriately allocating the number of quantization bits over them under a total rate constraint. Last, the quantized bits by RRHs are sent to the BBU for joint information decoding via their fronthaul links. It is worth noting that at each RRH, the compression of the correlated signals received by all its antennas is fulfilled by a simple linear filter, rather than the complicated distributed source coding, e.g., Wyner-Ziv coding, as in single-antenna C-RAN where all the antennas are physically separated \cite{Yuwei13,Yu13}.

\item With the proposed SCF scheme, we formulate the optimization problem of joint users' power allocation, RRHs' spatial filter design and quantization bits allocation, as well as BBU's receive beamforming to maximize the minimum signal-to-interference-plus-noise ratio (SINR) of all the users in the uplink. We propose an efficient solution to this complicated design problem for ease of practical implementation. First, each RRH computes its own spatial filter in a distributed manner based on its locally received signal covariance matrix. This distributed solution helps save significant fronthaul resources for exchanging control signals with the BBU. Then, the BBU jointly optimizes the other parameters to maximize the minimum SINR of all users based on the alternating optimization technique. Specifically, with given quantization bits allocation at each RRH, we extend the well-known fixed-point method \cite{Yates95} to obtain the optimal transmit power levels for the users and receive beamforming vectors at the BBU. On the other hand, given the above optimized parameters, the corresponding quantization bits allocation at each RRH is efficiently solved.

\item Last, we investigate the following interesting question: \emph{Given a total number of antennas for a target area, should all of them be deployed at one BS, i.e., massive MIMO, or one at each RRH, i.e., single-antenna C-RAN, or optimally divided over a certain number of RRHs, i.e., multi-antenna C-RAN, given a practical total fronthaul capacity constraint for C-RAN?} Through numerical examples, it is shown that with the proposed SCF scheme and joint wireless-fronthaul-BBU optimization, multi-antenna C-RAN generally performs noticeably better than both massive MIMO and single-antenna C-RAN under moderate fronthaul rate constraints.

\end{itemize}

The rest of this paper is organized as follows. Section \ref{sec:System Model} presents the system model for C-RAN with the SCF scheme. Section \ref{sec:Problem Formulation} formulates the minimum SINR maximization problem. Section \ref{sec:Proposed Solution} presents the proposed solution for this problem. Section \ref{sec:Numerical Results} provides numerical results to verify the effectiveness of the proposed multi-antenna C-RAN with SCF. Finally, Section \ref{sec:Conclusions} concludes the paper.

{\it Notation}: Scalars are denoted by lower-case letters, vectors
by bold-face lower-case letters, and matrices by
bold-face upper-case letters. $\mv{I}$ and $\mv{0}$  denote an
identity matrix and an all-zero matrix, respectively, with
appropriate dimensions. For a square matrix $\mv{S}$, $\mv{S}\succeq\mv{0}$ ($\mv{S}\preceq \mv{0}$) means that $\mv{S}$ is positive (negative) semi-definite. For a matrix
$\mv{M}$ of arbitrary size, $\mv{M}^{H}$ and ${\rm rank}(\mv{S})$ denote the
conjugate transpose and rank of $\mv{M}$, respectively. ${\rm diag}(\mv{a})$ denotes a diagonal matrix with the main diagonal given by the vector $\mv{a}$, while ${\rm diag}(\mv{A}_1,\cdots,\mv{A}_K)$ denotes a block diagonal matrix with the diagonal given by square matrices $\mv{A}_1,\cdots,\mv{A}_K$. $E[\cdot]$ denotes the statistical expectation. The
distribution of a circularly symmetric complex Gaussian (CSCG) random vector with mean $\mv{x}$ and
covariance matrix $\mv{\Sigma}$ is denoted by
$\mathcal{CN}(\mv{x},\mv{\Sigma})$; and $\sim$ stands for
``distributed as''. $\mathbb{C}^{x \times y}$ denotes the space of
$x\times y$ complex matrices. $\|\mv{x}\|$ denotes the Euclidean norm of a complex vector
$\mv{x}$. For two real vectors
$\mv{x}$ and $\mv{y}$, $\mv{x}\geq \mv{y}$ means that $\mv{x}$ is
greater than or equal to $\mv{y}$ in a component-wise manner.

\section{System Model}\label{sec:System Model}
This paper considers the uplink communication in a multi-antenna C-RAN. As shown in Fig. \ref{fig8}(c), the system consists of one BBU, $N$ RRHs, denoted by the set $\mathcal{N}=\{1,\cdots,N\}$, and $K$ users, denoted by the set $\mathcal{K}=\{1,\cdots,K\}$. It is assumed that each RRH is equipped with $M\geq 1$ antennas, while each user is equipped with one single antenna. It is further assumed that each RRH $n$ is connected to the BBU via a digital error-free fronthaul link with a capacity of $\bar{T}_n$ bits per second (bps). In the uplink, each RRH processes the signals received from all the users into digital bits and forwards them to the BBU via its fronthaul link. Then, the BBU jointly decodes the users' messages based on the signals from all the RRHs. The details of our studied uplink multi-antenna C-RAN are given as follows.

It is assumed that all the $K$ users transmit over quasi-static flat-fading channels over a given bandwidth of $B$ Hz.\footnote{It is worth noting that this paper studies the narrowband wireless channel, while the results obtained can be readily extended to broadband based C-RAN by viewing each frequency subchannel as one norrowband channel considered here and accordingly scaling down the fronthaul capacity $\bar{T}_n$ by the number of frequency subchannels in use.} The equivalent baseband complex symbol received at RRH $n$ is then expressed as
\begin{align}\label{eqn:received signal}
\mv{y}_n=\sum\limits_{k=1}^K\mv{h}_{n,k}\sqrt{p}_ks_k+\mv{z}_n=\mv{H}_n\mv{P}^{\frac{1}{2}}\mv{s}+\mv{z}_n, ~~~ n=1,\cdots,N,
\end{align}where $\mv{s}=[s_1,\cdots,s_K]^T$ with $s_k\sim \mathcal{CN}(0,1)$ denoting the transmit symbol of user $k$ which is modelled as a CSCG random variable with zero-mean and unit-variance, $\mv{P}={\rm diag}(p_1,\cdots,p_K)$ with $p_k$ denoting the transmit power of user $k$, $\mv{H}_n=[\mv{h}_{n,1},\cdots,\mv{h}_{n,K}]\in \mathbb{C}^{M\times K}$ with $\mv{h}_{n,k}\in \mathbb{C}^{M\times 1}$ denoting the channel vector from user $k$ to RRH $n$, and $\mv{z}_n\sim \mathcal{CN}(\mv{0},\sigma_n^2\mv{I})$ denotes the additive white Gaussian noise (AWGN) at RRH $n$. It is assumed that $\mv{z}_n$'s are independent over $n$.

\begin{figure}
\begin{center}
\scalebox{0.5}{\includegraphics*{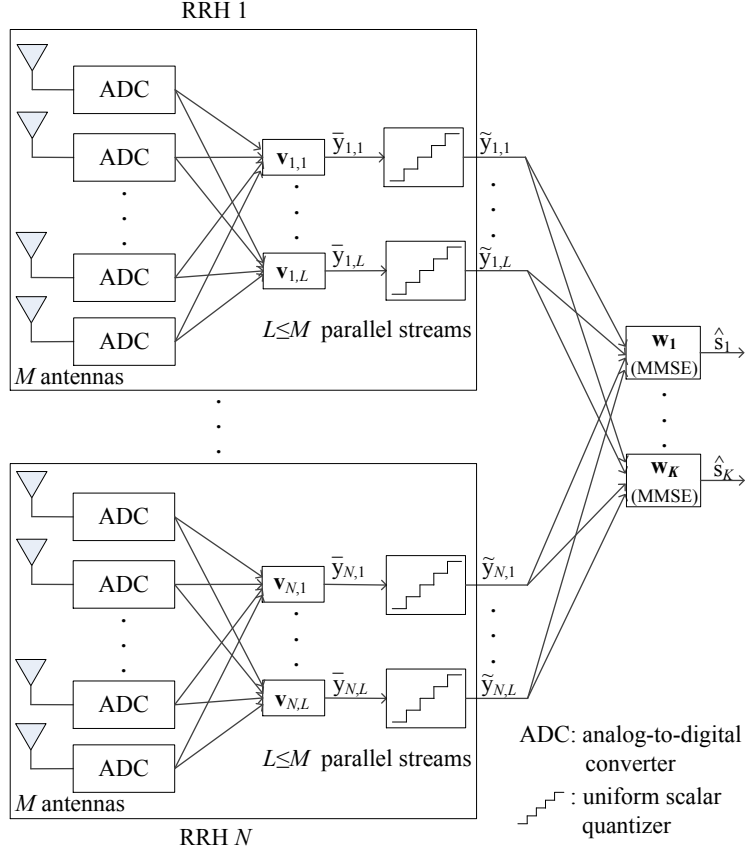}}
\end{center}
\caption{Illustration of C-RAN with SCF applied at each multi-antenna RRH}\label{fig2}
\end{figure}

At each multi-antenna RRH, we consider the use of the SCF scheme to forward the received baseband symbol $\mv{y}_n$ to the BBU via the finite-capacity digital fronthaul link. An illustration of the SCF scheme is shown in Fig. \ref{fig2}, where each RRH first demodulates the signal received from each antenna to the baseband and processes the baseband signals from all antennas via spatial filtering, then conducts a scalar quantization over each output in parallel, and finally forwards the total quantized bits to the BBU via the fronthaul link. Specifically, RRH $n$ first implements spatial filtering to its received signal $\mv{y}_n$ to obtain
\begin{align}\label{eqn:beamforming and quantize}
\bar{\mv{y}}_n=\mv{V}_n\mv{y}_n=\mv{V}_n\sum\limits_{k=1}^K\mv{h}_{n,k}\sqrt{p}_ks_k+\mv{V}_n\mv{z}_n=\mv{V}_n\mv{H}_n\mv{P}^{\frac{1}{2}}\mv{s}+\mv{V}_n\mv{z}_n, ~~~ n=1,\cdots,N,
\end{align}where $\mv{V}_n=[\mv{v}_{n,1},\cdots,\mv{v}_{n,L}]^T\in \mathbb{C}^{L\times M}$ denotes the filtering matrix at RRH $n$ with $L\leq M$ denoting the (reduced) dimension of the output signal $\bar{\mv{y}}_n$ after filtering, which will be specified later. Then, a simple uniform scalar quantization is applied to each element of $\bar{\mv{y}}_n=[\bar{y}_{n,1},\cdots,\bar{y}_{n,L}]^T$ at RRH $n$. Note that each complex symbol $\bar{y}_{n,l}$ can be represented by its in-phase (I) and quadrature (Q) parts as
\begin{align}\label{eqn:inphase and Q}
\bar{y}_{n,l}=\bar{y}_{n,l}^I+j\bar{y}_{n,l}^Q,  ~~~ \forall n,l,
\end{align}where $j^2=-1$, and the I-branch symbol $\bar{y}_{n,l}^I$ and Q-branch symbol $\bar{y}_{n,l}^Q$ are both real Gaussian random variables with zero mean and variance $(\sum_{k=1}^Kp_k|\mv{v}_{n,l}^T\mv{h}_{n,k}|^2+\sigma^2\|\mv{v}_{n,l}\|^2)/2$. A typical method to implement the uniform quantization is via separate I/Q scalar
quantization, the details of which can be found in \cite{Liu14}. After uniform scalar quantization, the baseband quantized symbol of $\bar{\mv{y}}_n$ is then given by
\begin{align}\label{eqn:quantized signal}
\tilde{\mv{y}}_n=\bar{\mv{y}}_n+\mv{e}_n=\mv{V}_n\sum\limits_{k=1}^K\mv{h}_{n,k}\sqrt{p}_ks_k+\mv{V}_n\mv{z}_n+\mv{e}_n, ~~~ n=1,\cdots,N,
\end{align}where $\mv{e}_n=[e_{n,1},\cdots,e_{n,L}]^T$ with $e_{n,l}$, $1\leq l \leq L$, denoting the quantization error for $\bar{y}_{n,l}$ with zero mean and variance $q_{n,l}$. Let $D_{n,l}$ denote the number of bits that RRH $n$ uses to quantize the I-branch or Q-branch of $\bar{y}_{n,l}$, i.e., $\bar{y}_{n,l}^I$ or $\bar{y}_{n,l}^Q$. The quantization noise level $q_{n,l}$ due to $e_{n,l}$ for uniform quantization is then given by \cite{Liu14}
\begin{align}\label{eqn:quantization noise}
q_{n,l}=\left\{\begin{array}{ll}3\left(\sum\limits_{k=1}^Kp_k|\mv{v}_{n,l}^T\mv{h}_{n,k}|^2+\sigma^2\|\mv{v}_{n,l}\|^2\right)2^{-2D_{n,l}}, & {\rm if} ~ D_{n,l}>0, \\ \infty, & {\rm if} ~ D_{n,l}=0, \end{array} \right. ~~~ l=1,\cdots,L, ~ n=1,\cdots,N.
\end{align}Note that $e_{n,l}$'s are independent over $l$ due to independent scalar quantization for each element of $\bar{\mv{y}}_n$, and also over $n$ due to independent processing at different RRHs. As a result, the covariance matrix of $\mv{e}_n$ is a function of $\mv{p}=[p_1,\cdots,p_K]^T$, $\mv{V}_n$ as well as $\mv{D}_n=[D_{n,1},\cdots,D_{n,L}]^T$, which is given by
\begin{align}
\mv{Q}_n(\mv{p},\mv{V}_n,\mv{D}_n)=E[\mv{e}_n\mv{e}_n^H]={\rm diag}(q_{n,1},\cdots,q_{n,L}).\end{align}

Then, each RRH forwards the quantized bits to the BBU via the fronthaul link. The transmission rate in RRH $n$'s fronthaul link is expressed as \cite{Liu14}
\begin{align}\label{eqn:fronthaul link}
T_n=2B\sum\limits_{l=1}^{L}D_{n,l}, ~~~ n=1,\cdots,N.
\end{align}It is worth noting that if $D_{n,l}=0$, then the symbol $\bar{y}_{n,l}$ at the spatial filter output dimension $l$ is not quantized at RRH $n$ and thus not forwarded to the BBU; as a result, $q_{n,l}=\infty$ according to (\ref{eqn:quantization noise}). To summarize, our proposed SCF scheme is a two-stage compression method, where the output signal dimension at each RRH is first reduced from $M$ to $L$ (if $M>L$) by spatial filtering, and then further reduced by allocating the limited fronthaul quantization bits to only a selected subset of the $L$ output dimensions with $D_{n,l}>0$.

The received signal at the BBU from all RRHs is expressed as
\begin{align}
\tilde{\mv{y}}=[\tilde{\mv{y}}_1^T,\cdots,\tilde{\mv{y}}_N^T]^T=\mv{V}\sum\limits_{k=1}^K\mv{h}_{(k)}\sqrt{p}_ks_k+\mv{V}\mv{z}+\mv{e},
\end{align}where $\mv{V}={\rm diag}(\mv{V}_1,\cdots,\mv{V}_N)$, $\mv{h}_{(k)}=[\mv{h}_{1,k}^T,\cdots,\mv{h}_{N,k}^T]^T$, $\mv{z}=[\mv{z}_1^T,\cdots,\mv{z}_N^T]^T$, and $\mv{e}=[\mv{e}_1^T,\cdots,\mv{e}_N^T]^T$. To decode $s_k$, we consider that a linear beamforming\footnote{It is worth noting that the subsequent results also hold similarly for the case with the non-linear successive interference cancellation scheme with any fixed decoding order, which is not considered in this paper due to its high complexity for implementation.} is applied to $\tilde{\mv{y}}$, i.e.,
\begin{align}\label{eqn:decoding}
\hat{s}_k=\mv{w}_k^H\tilde{\mv{y}}=\mv{w}_k^H\mv{V}\mv{h}_{(k)}\sqrt{p}_ks_k+\sum\limits_{j\neq k}\mv{w}_k^H\mv{V}\mv{h}_{(j)}\sqrt{p}_js_j+\mv{w}_k^H\mv{V}\mv{z}+\mv{w}_k^H\mv{e}, ~~~ k=1,\cdots,K,
\end{align}where $\mv{w}_k=[w_{k,1},\cdots,w_{k,NL}]^T\in \mathbb{C}^{NL\times 1}$ is the receive beamformer for $s_k$. Note that $w_{k,(n-1)L+l}=0$ if the $l$th signal dimension at RRH $n$ is not quantized, i.e., $D_{n,l}=0$. According to (\ref{eqn:decoding}), the SINR for decoding $s_k$ is expressed as
\begin{align}\label{eqn:SINR}
\gamma_k=\frac{p_k|\mv{w}_k^H\mv{V}\mv{h}_{(k)}|^2}{\sum\limits_{j\neq k}p_j|\mv{w}_k^H\mv{V}\mv{h}_{(j)}|^2+\sigma^2\|\mv{w}_k^H\mv{V}\|^2+\mv{w}_k^H\mv{Q}(\mv{p},\{\mv{V}_n,\mv{D}_n\})\mv{w}_k}, ~~~ k=1,\cdots,K,
\end{align}where $\mv{Q}(\mv{p},\{\mv{V}_n,\mv{D}_n\})=E[\mv{e}\mv{e}^H]={\rm diag}(\mv{Q}_1(\mv{p},\mv{V}_1,\mv{D}_1),\cdots,\mv{Q}_N(\mv{p},\mv{V}_N,\mv{D}_N))$.

\section{Problem Formulation}\label{sec:Problem Formulation}
%
%

In this paper, we aim to maximize the minimum SINR of all the users by optimizing the users' power allocation, i.e., $\mv{p}$, the spatial filter and quantization bits allocation at each RRH, i.e., $\mv{V}_n$ and $\mv{D}_n$, $\forall n$, and the receive beamforming vectors at the BBU, i.e., $\mv{w}_k$, $\forall k$. Specifically, we aim to solve the following problem:
\begin{align}\mathop{\mathtt{Maximize}}_{\gamma,\mv{p},\{\mv{V}_n,\mv{D}_n\},\{\mv{w}_k\}}
& ~~~ \gamma \label{eqn:problem 1} \\
\mathtt {Subject \ to} & ~~~ \frac{p_k|\mv{w}_k^H\mv{V}\mv{h}_{(k)}|^2}{\sum\limits_{j\neq k}p_j|\mv{w}_k^H\mv{V}\mv{h}_{(j)}|^2+\sigma^2\|\mv{w}_k^H\mv{V}\|^2+\mv{w}_k^H\mv{Q}(\mv{p},\{\mv{V}_n,\mv{D}_n\})\mv{w}_k} \geq \gamma, ~~~ \forall k, \label{eqn:SINR balancing 1}\\ & ~~~ p_k\leq \bar{P}_k, ~~~ \forall k, \label{eqn:power constraint 1} \\
& ~~~ 2B\sum\limits_{l=1}^LD_{n,l}\leq \bar{T}_n, ~~~ \forall n, \label{eqn:fronthaul 1} \\ & ~~~ D_{n,l} ~ {\rm is ~ an ~ integer}, ~~~ \forall n,l, \label{eqn:integer 1}
\end{align}where $\gamma$ denotes the common SINR target for all the $K$ users and $\bar{P}_k$ denotes the transmit power constraint of user $k$.

It can be observed that problem (\ref{eqn:problem 1}) is a non-convex optimization problem since the design variables are complicatedly coupled in its constraint (\ref{eqn:SINR balancing 1}). Note that with given $\mv{V}_n$'s and by setting $D_{n,l}=\infty$, $\forall n,l$, problem (\ref{eqn:problem 1}) reduces to the well-known SINR balancing problem via transmit power control and receive beamforming only. This problem has been efficiently solved by \cite{Schubert04,Zhanglan08} based on the non-negative matrix theory \cite{Horn85}. However, due to the new quantization noise term in the SINR expression, i.e., $\mv{w}_k^H\mv{Q}(\mv{p},\{\mv{V}_n,\mv{D}_n\})\mv{w}_k$, it can be shown that the algorithm proposed in \cite{Schubert04,Zhanglan08} cannot be directly applied to obtain the optimal power control and beamforming solution to problem (\ref{eqn:problem 1}) even with given $\mv{V}_n$'s and finite values of $D_{n,l}$'s. As a result, with the additional variables $\mv{V}_n$'s and integer variables $\mv{D}_n$'s, it is generally difficult to globally solve problem (\ref{eqn:problem 1}).

\begin{figure}
\begin{center}
\scalebox{0.5}{\includegraphics*{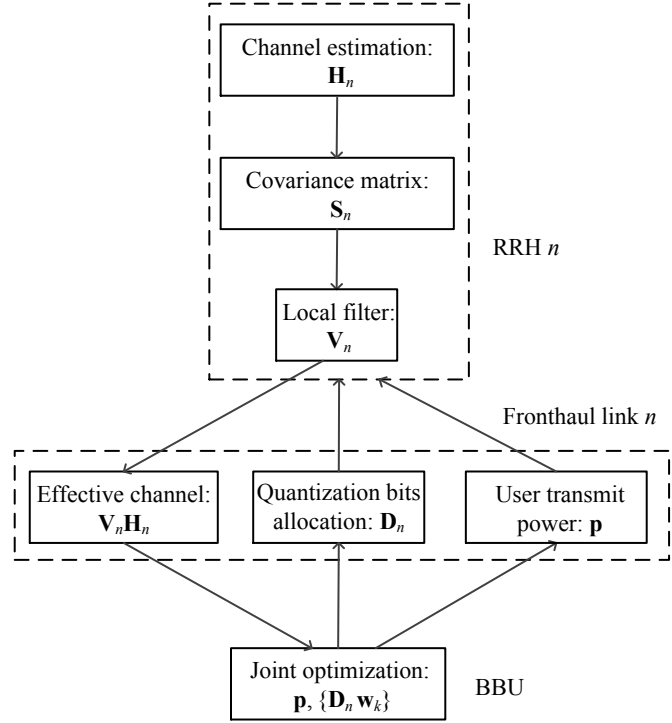}}
\end{center}
\caption{A diagram for illustrating the proposed algorithm for Problem (\ref{eqn:problem 1}).}\label{fig10}
\end{figure}

In this paper, we propose an efficient algorithm to solve problem (\ref{eqn:problem 1}) suboptimally, for which the main procedures are illustrated in Fig. \ref{fig10}. Specifically, in the channel training phase, all the users send orthogonal pilot signals to the RRHs with equal power $\bar{P}\leq \min_{1\leq k \leq K} \bar{P}_k$. Let $\mv{y}_{pilot,n}\in \mathbb{C}^{M\times 1}$ denote the received signal at RRH $n$ during the channel training phase. Then, each RRH $n$ estimates the MIMO channel $\mv{H}_n$ (assumed to be perfect) and computes the covariance matrix of its received signal (with the presumed knowledge of $\bar{P}$ and $\sigma^2$), i.e.,
\begin{align}\label{eqn:covariance matrix}\mv{S}_n=E[\mv{y}_{pilot,n}\mv{y}_{pilot,n}^H]=\mv{H}_n\mv{P}\mv{H}_n^H+\sigma^2\mv{I}=\bar{P}\mv{H}_n\mv{H}_n^H+\sigma^2\mv{I}, ~~~ n=1,\cdots,N.
\end{align}Then each RRH independently determines its spatial filter $\mv{V}_n$ based on $\mv{S}_n$. Next, each RRH $n$ sends the effective channel $\mv{V}_n\mv{H}_n\triangleq \tilde{\mv{H}}_n=[\tilde{\mv{h}}_{n,1},\cdots,\tilde{\mv{h}}_{n,K}]$ to the BBU. Then, based on its global CSI, i.e., $\tilde{\mv{H}}_n$'s, $n=1,\cdots,N$, the BBU computes $\mv{p}$, $\mv{D}_n$'s and $\mv{w}_k$'s by solving the following reduced problem of problem (\ref{eqn:problem 1}):
\begin{align}\mathop{\mathtt{Maximize}}_{\gamma,\mv{p},\{\mv{D}_n\},\{\mv{w}_k\}}
& ~~~ \gamma \label{eqn:problem 11} \\
\mathtt {Subject \ to} & ~~~ \frac{p_k|\mv{w}_k^H\tilde{\mv{h}}_{(k)}|^2}{\sum\limits_{j\neq k}p_j|\mv{w}_k^H\tilde{\mv{h}}_{(j)}|^2+\sigma^2\|\mv{w}_k^H\mv{V}\|^2+\mv{w}_k^H\mv{Q}(\mv{p},\{\mv{V}_n,\mv{D}_n\})\mv{w}_k} \geq \gamma, ~~~ \forall k, \label{eqn:SINR balancing 11}\\ & ~~~ p_k\leq \bar{P}_k, ~~~ \forall k, \label{eqn:power constraint 11} \\
& ~~~ 2B\sum\limits_{l=1}^LD_{n,l}\leq \bar{T}_n, ~~~ \forall n, \label{eqn:fronthaul 11} \\ & ~~~ D_{n,l} ~ {\rm is ~ an ~ integer}, ~~~ \forall n,l, \label{eqn:integer 11}
\end{align}where $\tilde{\mv{h}}_{(k)}=[\tilde{\mv{h}}_{1,k}^T,\cdots,\tilde{\mv{h}}_{N,k}^T]^T$. Finally, the BBU sends $\mv{D}_n$ to RRH $n$, $n=1,\cdots,N$, via its fronthaul link for the implementation of scalar quantization, and $\mv{p}$ to one RRH (say RRH $n$), which then broadcasts $\mv{p}$ to all the users in the downlink. In the following section, we first show how each RRH computes its spatial filter in a distributed manner, and then discuss how the BBU solves problem (\ref{eqn:problem 11}).

\section{Proposed Algorithm for Problem (\ref{eqn:problem 1})}\label{sec:Proposed Solution}
In this section, we present details of the two-stage algorithm for solving problem (\ref{eqn:problem 1}) based on the diagram shown in Fig. \ref{fig10}.

\subsection{Distributed Spatial Filter Design at RRHs}
First, we show how to obtain the spatial filter at each RRH. It is observed from problem (\ref{eqn:problem 1}) that each user's SINR at the BBU depends on $\mv{V}_n$'s, $n=1,\cdots,N$, at all RRHs, which are also coupled with other design variables. As a result, it is difficult to obtain the jointly optimal spatial filters at all RRHs. Therefore, in the following we propose to design $\mv{V}_n$'s in a distributed manner, where each RRH $n$ computes its own spatial filter only based on its local information, i.e., $\mv{S}_n$ in (\ref{eqn:covariance matrix}). First, each RRH $n$ obtains the eigenvalue decomposition (EVD) of $\mv{S}_n$ as
\begin{align}
\mv{S}_n=\mv{U}_n^H\mv{\Lambda}_n\mv{U}_n, ~~~ n=1,\cdots,N,
\end{align}where $\mv{\Lambda}_n={\rm diag}(\lambda_{n,1},\cdots,\lambda_{n,M})\in \mathbb{C}^{M\times M}$ with $\lambda_{n,1}\geq \cdots \geq \lambda_{n,M} \geq 0$ denoting the eigenvalues of $\mv{S}_n$, and $\mv{U}_n=[\mv{u}_{n,1},\cdots,\mv{u}_{n,M}]\in \mathbb{C}^{M\times M}$ is a unitary matrix with $\mv{u}_{n,m}$ denoting the eigenvector of $\mv{S}_n$ associated with the eigenvalue $\lambda_{n,m}$, $m=1,\cdots,M$. Notice that $\mv{U}_n$ can also be obtained as the singular vectors of $\mv{H}_n$ from its singular value decomposition (SVD). Note that since the rank of $\mv{S}_n$ is ${\rm rank}(\mv{S}_n)=\min(M,K)$, $\lambda_{n,m}=0$ if $\min(M,K)<m\leq M$, and thus each RRH only needs to extract the first $L=\min(M,K)$ user signal subspace by setting its spatial filter as 
\begin{align}\label{eqn:beamforming at RRH}
\mv{V}_n=\bar{\mv{U}}_n\triangleq [\mv{u}_{n,1},\cdots,\mv{u}_{n,L}], ~~~ n=1,\cdots,N.
\end{align}As a result, the equivalent channel from all users to RRH $n$ after filtering is given by
\begin{align}
\tilde{\mv{H}}_n=\bar{\mv{U}}_n\mv{H}_n, ~~~ n=1,\cdots,N,
\end{align}which is sent to the BBU in channel training as discussed in Section \ref{sec:Problem Formulation}.

The reasons to set the spatial filter at each RRH as in (\ref{eqn:beamforming at RRH}) are as follows. First, if $M>K$, $\bar{\mv{U}}_n$ denoises the received signal before quantization by extracting the user signal subspace from the received signal plus noise without loss of information \cite{Distributed06,TomLuo07}. Second, if $M>K$, the signal dimension is reduced from $M$ to $L=K$ at each RRH $n$, which simplifies the subsequent uniform scalar quantization at each dimension. Last but not least, with this spatial filer applied at RRH $n$ in data transmission, the filter output $\bar{\mv{y}}_n$ given in (\ref{eqn:beamforming and quantize}) has the following covariance matrix,
\begin{align}
E[\bar{\mv{y}}_n\bar{\mv{y}}_n^H]=\bar{\mv{U}}_n\mv{H}_n\mv{P}\mv{H}_n^H\bar{\mv{U}}_n^H+\sigma^2\mv{I},
\end{align}which is (approximately) diagonal if the transmit power of users, i.e., $p_k$'s, are (closely) equal, i.e., the elements in $\bar{\mv{y}}_n$ are approximately uncorrelated. As a result, independent scalar quantization at each output dimension of $\bar{\mv{y}}_n$ as shown in (\ref{eqn:quantized signal}) becomes more effective.

\subsection{Joint Optimization at BBU}
Next, given the spatial filter at each RRH $\mv{V}_n=\bar{\mv{U}}_n$, $\forall n$, as in (\ref{eqn:beamforming at RRH}), we jointly optimize users' power allocation, RRHs' quantization bits allocation, and BBU's receive beamforming by solving problem (\ref{eqn:problem 11}).

We propose to apply the alternating optimization technique to solve problem (\ref{eqn:problem 11}). Specifically, first we fix the quantization bits allocation $D_{n,l}=\bar{D}_{n,l}$, $\forall n,l$, in problem (\ref{eqn:problem 11}) and optimize the users' power allocation and BBU's receive beamforming solution by solving the following problem.
\begin{align}\mathop{\mathtt{Maximize}}_{\gamma,\mv{p},\{\mv{w}_k\}}
& ~~~ \gamma \label{eqn:problem 2} \\
\mathtt {Subject \ to} & ~~~ \frac{p_k|\mv{w}_k^H\tilde{\mv{h}}_{(k)}|^2}{\sum\limits_{j\neq k}p_j|\mv{w}_k^H\tilde{\mv{h}}_{(j)}|^2+\sigma^2\|\mv{w}_k\|^2+\mv{w}_k^H\mv{Q}(\mv{p},\{\bar{\mv{U}}_n,\bar{\mv{D}}_n\})\mv{w}_k} \geq \gamma, ~~~ \forall k, \label{eqn:SINR balancing 2} \\ & ~~~ p_k\leq \bar{P}_k, ~~~ \forall k. \label{eqn:power constraint 2}
\end{align}Let $\bar{\mv{p}}=[\bar{p}_1,\cdots,\bar{p}_K]^T$ and $\bar{\mv{w}}_k$'s denote the solution to problem (\ref{eqn:problem 2}). Then, we fix the users' power allocation and BBU's receive beamforming solution as $\mv{p}=\bar{\mv{p}}$, $\mv{w}_k=\bar{\mv{w}}_k$'s, respectively, and optimize the RRHs' quantization bits allocation by solving the following problem.
\begin{align}\mathop{\mathtt{Maximize}}_{\gamma,\{\mv{D}_n\}}
& ~~~ \gamma \label{eqn:problem 3} \\
\mathtt {Subject \ to} & ~~~ \frac{\bar{p}_k|\bar{\mv{w}}_k^H\tilde{\mv{h}}_{(k)}|^2}{\sum\limits_{j\neq k}\bar{p}_j|\bar{\mv{w}}_k^H\tilde{\mv{h}}_{(j)}|^2+\sigma^2\|\bar{\mv{w}}_k\|^2+\bar{\mv{w}}_k^H\mv{Q}(\bar{\mv{p}},\{\bar{\mv{U}}_n,\mv{D}_n\})\bar{\mv{w}}_k} \geq \gamma, ~~~ \forall k, \label{eqn:SINR balancing 3}\\
& ~~~ 2B\sum\limits_{l=1}^LD_{n,l}\leq \bar{T}_n, ~~~ \forall n, \label{eqn:fronthaul 3} \\ & ~~~ D_{n,l} ~ {\rm is ~ an ~ integer}, ~~~ \forall n,l. \label{eqn:integer 3}
\end{align}Let $\bar{\mv{D}}_n$'s denote the solution to problem (\ref{eqn:problem 3}). The above procedure is iterated until convergence. In the following, we show how to solve problems (\ref{eqn:problem 2}) and (\ref{eqn:problem 3}), respectively.

{\it 1. Solution to Problem (\ref{eqn:problem 2})}

First, we solve problem (\ref{eqn:problem 2}). As mentioned in Section \ref{sec:Problem Formulation}, due to the additional term accociated with the quantization noise, i.e., $\mv{w}_k^H\mv{Q}(\mv{p},\{\bar{\mv{U}}_n,\bar{\mv{D}}_n\})\mv{w}_k$'s, in the SINR expression, which is coupled with $\mv{p}$ and $\mv{w}_k$'s, the algorithm proposed in \cite{Schubert04} and \cite{Zhanglan08} cannot be applied to obtain the optimal solution to problem (\ref{eqn:problem 2}). In the following, we propose to apply the bisection method \cite{Boyd04} together with the celebrated fixed-point method \cite{Yates95} to globally solve problem (\ref{eqn:problem 2}). Specifically, given any common SINR target $\bar{\gamma}$ for all the users, we solve the following feasibility problem.
\begin{align}\mathop{\mathtt{Find}}
& ~~~ \mv{p},\{\mv{w}_k\} \label{eqn:problem 4} \\
\mathtt {Subject \ to} & ~~~ \frac{p_k|\mv{w}_k^H\tilde{\mv{h}}_{(k)}|^2}{\sum\limits_{j\neq k}p_j|\mv{w}_k^H\tilde{\mv{h}}_{(j)}|^2+\sigma^2\|\mv{w}_k\|^2+\mv{w}_k^H\mv{Q}(\mv{p},\{\bar{\mv{U}}_n,\bar{\mv{D}}_n\})\mv{w}_k} \geq \bar{\gamma}, ~~~ \forall k, \label{eqn:SINR balancing 4} \\ & ~~~ p_k\leq \bar{P}_k, ~~~ \forall k. \label{eqn:power constraint 4}
\end{align}The feasibility problem (\ref{eqn:problem 4}) is still a non-convex problem. In the following, we show that the feasibility of problem (\ref{eqn:problem 4}) can be efficiently checked by solving a sum-power minimization problem without the individual power constraints in (\ref{eqn:power constraint 4}). The sum-power minimization problem is formulated as
\begin{align}\mathop{\mathtt{Minimize}}_{\mv{p},\{\mv{w}_k\}}
& ~~~ \sum\limits_{k=1}^K p_k \label{eqn:problem 5} \\
\mathtt {Subject \ to} & ~~~ \frac{p_k|\mv{w}_k^H\tilde{\mv{h}}_{(k)}|^2}{\sum\limits_{j\neq k}p_j|\mv{w}_k^H\tilde{\mv{h}}_{(j)}|^2+\sigma^2\|\mv{w}_k\|^2+\mv{w}_k^H\mv{Q}(\mv{p},\{\bar{\mv{U}}_n,\bar{\mv{D}}_n\})\mv{w}_k} \geq \bar{\gamma}, ~~~ \forall k. \label{eqn:SINR balancing 5}
\end{align}Note that given any power allocation $\mv{p}$, the optimal linear receive beamforming solution to problem (\ref{eqn:problem 5}) is the minimum-mean-square-error (MMSE) based receiver given by
\begin{align}\label{eqn:MMSE}
\mv{w}_k^{{\rm MMSE}}=\left(\sum\limits_{j\neq k}p_j\tilde{\mv{h}}_{(j)}\tilde{\mv{h}}_{(j)}^H+\sigma^2\mv{I}+\mv{Q}(\mv{p},\{\bar{\mv{U}}_n,\bar{\mv{D}}_n\})\right)^{-1}\tilde{\mv{h}}_{(k)}, ~~~ k=1,\cdots,K.
\end{align}With the above MMSE receivers, the SINR of user $k$ given in (\ref{eqn:SINR}) reduces to
\begin{align}\label{eqn:optimal SINR}
\gamma_k=p_k\tilde{\mv{h}}_{(k)}^H\left(\sum\limits_{j\neq k}p_j\tilde{\mv{h}}_{(j)}\tilde{\mv{h}}_{(j)}^H+\sigma^2\mv{I}+\mv{Q}(\mv{p},\{\bar{\mv{U}}_n,\bar{\mv{D}}_n\})\right)^{-1}\tilde{\mv{h}}_{(k)}, ~~~ k=1,\cdots,K.
\end{align}As a result, problem (\ref{eqn:problem 5}) reduces to the following power control problem.
\begin{align}\mathop{\mathtt{Minimize}}_{\mv{p}}
& ~~~ \sum\limits_{k=1}^K p_k \label{eqn:problem 6} \\
\mathtt {Subject \ to} & ~~~ p_k\tilde{\mv{h}}_{(k)}^H\left(\sum\limits_{j\neq k}p_j\tilde{\mv{h}}_{(j)}\tilde{\mv{h}}_{(j)}^H+\sigma^2\mv{I}+\mv{Q}(\mv{p},\{\bar{\mv{U}}_n,\bar{\mv{D}}_n\})\right)^{-1}\tilde{\mv{h}}_{(k)} \geq \bar{\gamma}, ~~~ \forall k. \label{eqn:SINR balancing 6}
\end{align}

Define $\mv{I}(\mv{p})=[\mv{I}_1(\mv{p}),\cdots,\mv{I}_K(\mv{p})]^T\in \mathbb{R}^{K\times 1}$ with the $k$th element denoted by
\begin{align}\label{eqn:interference function}
\mv{I}_k(\mv{p})=\frac{\bar{\gamma}}{\tilde{\mv{h}}_{(k)}^H\left(\sum\limits_{j\neq k}p_j\tilde{\mv{h}}_{(j)}\tilde{\mv{h}}_{(j)}^H+\sigma^2\mv{I}+\mv{Q}(\mv{p},\{\bar{\mv{U}}_n,\bar{\mv{D}}_n\})\right)^{-1}\tilde{\mv{h}}_{(k)}}, ~~~ k=1,\cdots,K.
\end{align}Then, problem (\ref{eqn:problem 6}) reduces to the following problem.
\begin{align}\mathop{\mathtt{Minimize}}_{\mv{p}}
& ~~~ \sum\limits_{k=1}^K p_k \label{eqn:problem 7} \\
\mathtt {Subject \ to} & ~~~ \mv{p}\geq \mv{I}(\mv{p}). \label{eqn:SINR balancing 7}
\end{align}

\begin{lemma}\label{lemma1}
$\mv{I}(\mv{p})$ given in (\ref{eqn:interference function}) is a standard interference function \cite{Yates95}. In other words, it satisfies 1. $\mv{I}(\mv{p})\geq \mv{0}$ if $\mv{p}\geq \mv{0}$; 2. if $\mv{p}\geq \mv{p}'$, then $\mv{I}(\mv{p})\geq \mv{I}(\mv{p}')$; and 3. $\forall a>1$, it follows that $a\mv{I}(\mv{p})>\mv{I}(a\mv{p})$.
\end{lemma}

\begin{proof}
Please refer to Appendix \ref{appendix1}.
\end{proof}

Based on Lemma \ref{lemma1}, we have the following corollaries.

\begin{corollary}\label{corollary1}
If $\bar{\gamma}$ is not achievable for all the $K$ users in problem (\ref{eqn:problem 6}), the iterative fixed-point method $\mv{p}^{(i+1)}=\mv{I}(\mv{p}^{(i)})$ with the initial point $\mv{p}^{(0)}=\mv{0}$, where $\mv{p}^{(i)}$ denotes the power allocation solution obtained in the $i$th iteration of the above fixed-point method, will converge to $\mv{p}\rightarrow \infty$.
\end{corollary}

\begin{proof}
Please refer to Appendix \ref{appendix5}.
\end{proof}

\begin{corollary}\label{corollary3}
If $\bar{\gamma}$ is achievable for all the $K$ users in problem (\ref{eqn:problem 6}), the iterative fixed-point method $\mv{p}^{(i+1)}=\mv{I}(\mv{p}^{(i)})$ will converge to the optimal solution to problem (\ref{eqn:problem 7}) with any initial point $\mv{p}^{(0)}\geq \mv{0}$.
\end{corollary}

\begin{proof}
Please refer to \cite[Theorem 2]{Yates95}.
\end{proof}

\begin{corollary}\label{corollary2}
If $\bar{\gamma}$ is achievable for all the $K$ users in problem (\ref{eqn:problem 6}), the optimal power solution $\mv{p}^\ast$ to problem (\ref{eqn:problem 7}) is component-wise minimum in the sense that any other feasible power solution $\mv{p}'$ that satisfies (\ref{eqn:SINR balancing 7}) must satisfy $\mv{p}'\geq \mv{p}^\ast$.
\end{corollary}

\begin{proof}
Please refer to Appendix \ref{appendix2}.
\end{proof}

According to Corollaries \ref{corollary1}-\ref{corollary2}, the feasibility of problem (\ref{eqn:problem 4}) can be efficiently checked as follows. First, we apply the fixed-point method in Corollary \ref{corollary1} to solve the sum-power minimization problem (\ref{eqn:problem 7}) with the initial point $\mv{p}^{(0)}=\mv{0}$. If the obtained solution, denoted by $\mv{p}^\ast$, is infinity (unbounded), then $\bar{\gamma}$ cannot be achieved by all the $K$ users simultaneously even without the individual power constraints. Therefore, problem (\ref{eqn:problem 4}) is not feasible. Otherwise, if $\mv{p}^\ast$ is of finite value, $\bar{\gamma}$ can be achieved by all the users without the individual power constraints. In this case, we check whether $\mv{p}^\ast$ satisfies all the given individual power constraints. If $\mv{p}^\ast$ does not satisfy all the individual power constraints, according to Corollary \ref{corollary2}, all the power solutions that satisfy the SINR constraints cannot satisfy the given individual power constraints. As a result, problem (\ref{eqn:problem 4}) is not feasible. Otherwise, if $\mv{p}^\ast$ satisfies all the individual power constraints, it is a feasible solution to problem (\ref{eqn:problem 4}), i.e., problem (\ref{eqn:problem 4}) is feasible.

Let $\gamma^\ast$ denote the optimal value of problem (\ref{eqn:problem 2}). The algorithm to solve problem (\ref{eqn:problem 2}) based on the bisection method is summarized in Table \ref{table1}. Note that in Step 1, $\gamma_{{\rm max}}$ can be set as $\gamma_{{\rm max}}=\max_{1\leq k\leq K} \tilde{\gamma}_k$, where $\tilde{\gamma}_k$ is obtained by setting $p_k=\bar{P}_k$ and $p_j=0$, $\forall j\neq k$, in (\ref{eqn:optimal SINR}).
\begin{table}[htp]
\begin{center}
\caption{\textbf{Algorithm \ref{table1}}: Algorithm for
Problem (\ref{eqn:problem 2})} \vspace{0.2cm}
 \hrule
\vspace{0.2cm}
\begin{itemize}
\item[1.] {\bf Initialize} $\gamma_{{\rm min}}=0$, $\gamma_{{\rm max}}\geq \gamma^\ast$;
\item[2.] {\bf Repeat}
\begin{itemize}
\item[a)] Set $\bar{\gamma}=\frac{\gamma_{{\rm max}}+\gamma_{{\rm min}}}{2}$;
\item[b)] Solve problem (\ref{eqn:problem 7}) by the fixed-point method as in Corollary \ref{corollary1} with initial point $\mv{p}^{(0)}=\mv{0}$;
\item[c)] If the obtained solution $\mv{p}^\ast$ is infinity or does not satisfy all the given individual power constraints in (\ref{eqn:power constraint 2}), set $\gamma_{{\rm max}}=\bar{\gamma}$; otherwise, set $\gamma_{{\rm min}}=\bar{\gamma}$;
\end{itemize}
\item[3.] {\bf Until} $\gamma_{{\rm max}}-\gamma_{{\rm min}}\leq \epsilon$, where $\epsilon>0$ is a small constraint to control the algorithm accuracy;
\item[4.] Calculate $\mv{w}_k$'s by substituting $\mv{p}^\ast$ into (\ref{eqn:MMSE}).
\end{itemize}
\vspace{0.2cm} \hrule \label{table1}
\end{center}
\end{table}


{\it 2. Solution to Problem (\ref{eqn:problem 3})}

Next, we consider problem (\ref{eqn:problem 3}). Note that the quantization noise power given in (\ref{eqn:quantization noise}), i.e., $q_{n,l}$, is a continuous function over quantization bits allocation $D_{n,l}$ (assuming it is continuous), except when $D_{n,l}=0$. To deal with this issue, we approximate the quantization noise power with the following continuous function of $D_{n,l}$, i.e.,
\begin{align}\label{eqn:approaximated quantization noise}
q_{n,l}=3\left(\sum\limits_{k=1}^Kp_k|\mv{v}_{n,l}^T\mv{h}_{n,k}|^2+\sigma^2\|\mv{v}_{n,l}\|^2\right)2^{-2D_{n,l}}, ~~~ D_{n,l}\geq 0, ~~~ \forall n,l.
\end{align}The difference of (\ref{eqn:approaximated quantization noise}) from (\ref{eqn:quantization noise}) is that when $D_{n,l}=0$, $q_{n,l}$ equals $3(\sum_{k=1}^Kp_k|\mv{v}_{n,l}^T\mv{h}_{n,k}|^2+\sigma^2\|\mv{v}_{n,l}\|^2)$ instead of infinity. This indicates that even with $D_{n,l}=0$, some quantized information at the $l$th signal dimension of RRH $n$ is forwarded to the BBU. However, since $3(\sum_{k=1}^Kp_k|\mv{v}_{n,l}^T\mv{h}_{n,k}|^2+\sigma^2\|\mv{v}_{n,l}\|^2)$ is generally much larger than the signal power of any user $k$, i.e., $p_k|\mv{v}_{n,l}^T\mv{h}_{n,k}|^2$'s, its effect on the user SINR given in (\ref{eqn:SINR}) is negligible. Therefore, we use (\ref{eqn:approaximated quantization noise}) to approximate the quantization noise power for solving problem (\ref{eqn:problem 3}). Note that when the actual SINR is computed using (\ref{eqn:SINR}), the quantization noise power is set to be infinity if the obtained $D_{n,l}$ is zero.

With the approximation (\ref{eqn:approaximated quantization noise}), problem (\ref{eqn:problem 3}) is still challenging to solve due to the integer constraints for $D_{n,l}$'s. Let problem (\ref{eqn:problem 3}-NoInt) denote the relaxation of problem (\ref{eqn:problem 3}) without the integer constraints in (\ref{eqn:integer 3}). In the following, we first solve problem (\ref{eqn:problem 3}-NoInt), the solution to which may not satisfy all the integer constraints. Then, we propose an efficient algorithm to obtain a set of integer solutions for all $D_{n,l}$'s based on the solution of relaxed problem (\ref{eqn:problem 3}-NoInt).

Similar to problem (\ref{eqn:problem 2}), problem (\ref{eqn:problem 3}-NoInt) can be globally solved by the bisection method. Given any common SINR target $\bar{\gamma}$ for all the users, we need to solve the following feasibility problem.
\begin{align}\mathop{\mathtt{Find}}
& ~~~ \{\mv{D}_n\} \label{eqn:problem 8} \\
\mathtt {Subject \ to} & ~~~ \frac{\bar{p}_k|\bar{\mv{w}}_k^H\tilde{\mv{h}}_{(k)}|^2}{\sum\limits_{j\neq k}\bar{p}_j|\bar{\mv{w}}_k^H\tilde{\mv{h}}_{(j)}|^2+\sigma^2\|\bar{\mv{w}}_k\|^2+\bar{\mv{w}}_k^H\mv{Q}(\bar{\mv{p}},\{\bar{\mv{U}}_n,\mv{D}_n\})\bar{\mv{w}}_k} \geq \bar{\gamma}, ~~~ \forall k, \label{eqn:SINR balancing 8}\\
& ~~~ 2B\sum\limits_{l=1}^LD_{n,l}\leq \bar{T}_n, ~~~ \forall n. \label{eqn:fronthaul 8}
\end{align}Note that we have
\begin{align}
\bar{\mv{w}}_k^H\mv{Q}(\bar{\mv{p}},\{\bar{\mv{U}}_n,\mv{D}_n\})\bar{\mv{w}}_k=\sum\limits_{n=1}^N\sum\limits_{l=1}^Lq_{n,l}|\bar{w}_{k,(n-1)L+m}|^2=\sum\limits_{n=1}^N\sum\limits_{l=1}^L\theta_{n,l,k}2^{-2D_{n,l}}, ~~~ \forall k,
\end{align}where $\bar{w}_{i,j}$ denotes the $j$th element of $\bar{\mv{w}}_i$, $1\leq i\leq K$, $1\leq j \leq NL$, and \begin{align}\theta_{n,l,k}=3|\bar{w}_{k,(n-1)L+l}|^2\left(\sum\limits_{i=1}^Kp_i|\mv{h}_{n,i}|^2+\sigma^2\right), ~~~ \forall n,l.
\end{align}Note that $\theta_{n,l,k}$ can be interpreted as the effective quantization noise power due to the $l$th quantized dimension at RRH $n$ in decoding $s_k$ at the BBU. Then, problem (\ref{eqn:problem 8}) is equivalent to the following problem.
\begin{align}\mathop{\mathtt{Find}}
& ~~~ \{\mv{D}_n\} \label{eqn:problem 9} \\
\mathtt {Subject \ to} & ~~~ \sum\limits_{n=1}^N\sum\limits_{l=1}^L\theta_{n,l,k}2^{-2D_{n,l}}\leq \frac{\bar{p}_k|\bar{\mv{w}}_k^H\tilde{\mv{h}}_{(k)}|^2}{\bar{\gamma}}-\sum\limits_{j\neq k}\bar{p}_j|\bar{\mv{w}}_k^H\tilde{\mv{h}}_{(j)}|^2-\sigma^2\|\bar{\mv{w}}_k\|^2, ~~~ \forall k, \label{eqn:SINR balancing 9}\\
& ~~~ 2B\sum\limits_{l=1}^LD_{n,l}\leq \bar{T}_n, ~~~ \forall n. \label{eqn:fronthaul 9}
\end{align}Problem (\ref{eqn:problem 9}) can be shown to be a convex feasibility problem, and thus can be efficiently solved via the interior-point method \cite{Boyd04}. Let $\gamma^\star$ denote the optimal value of problem (\ref{eqn:problem 3}-NoInt). The algorithm for problem (\ref{eqn:problem 3}-NoInt) is then summarized in Table \ref{table2}. Note that in Step 1, $\gamma_{{\rm max}}$ can be similarly set as in Algorithm \ref{table1}.
\begin{table}[htp]
\begin{center}
\caption{\textbf{Algorithm \ref{table2}}: Algorithm for
Problem (\ref{eqn:problem 3}-NoInt)} \vspace{0.2cm}
 \hrule
\vspace{0.2cm}
\begin{itemize}
\item[1.] {\bf Initialize} $\gamma_{{\rm min}}=0$, $\gamma_{{\rm max}}\geq \gamma^\star$;
\item[2.] {\bf Repeat}
\begin{itemize}
\item[a)] Set $\bar{\gamma}=\frac{\gamma_{{\rm max}}+\gamma_{{\rm min}}}{2}$;
\item[b)] Solve problem (\ref{eqn:problem 9}) by the interior-point method. If problem (\ref{eqn:problem 9}) is feasible, set $\gamma_{{\rm min}}=\bar{\gamma}$; otherwise, set $\gamma_{{\rm max}}=\bar{\gamma}$;
\end{itemize}
\item[3.] {\bf Until} $\gamma_{{\rm max}}-\gamma_{{\rm min}}\leq \epsilon$, where $\epsilon>0$ is a small constraint to control the algorithm accuracy.
\end{itemize}
\vspace{0.2cm} \hrule \label{table2}
\end{center}
\end{table}

It is worth noting that the solution obtained by Algorithm \ref{table2} may not satisfy all the integer constraints given by (\ref{eqn:integer 3}) in problem (\ref{eqn:problem 3}). In the following, we show how to obtain a set of integer solutions $\{\mv{D}_n\}$ for problem (\ref{eqn:problem 3}) based on the solution of relaxed problem (\ref{eqn:problem 3}-NoInt). Similar to \cite{Liu14}, we propose to round each $D_{n,l}$ to its nearby integer as follows.
\begin{align}\label{eqn:feasible}
\hat{D}_{n,l}=\left\{\begin{array}{ll}\lfloor D_{n,l} \rfloor, & {\rm if} ~ D_{n,l}-\lfloor D_{n,l} \rfloor \leq \alpha_n, \\ \lceil D_{n,l} \rceil, & {\rm otherwise},\end{array}\right. ~~~ n=1,\cdots,N, ~~~ l=1,\cdots,L,
\end{align}where $0\leq \alpha_n \leq 1$, $\forall n$, and $\lfloor x \rfloor$, $\lceil x \rceil$ denote the maximum integer that is no larger than $x$ and the minimum integer that is no smaller than $x$, respectively. Note that we can always find a feasible solution of $\hat{D}_{n,l}$'s by simply setting $\alpha_n=1$, $\forall n$, in (\ref{eqn:feasible}) since in this case we have $\sum_{l=1}^L\hat{D}_{n,l}\leq \sum_{l=1}^LD_{n,l} \leq \bar{T}_n/2B$. In the following, we show how to find a generally better feasible solution by optimizing $\alpha_n$'s.
Since $\hat{D}_{n,l}$'s decrease with $\alpha_n$'s, as $\alpha_n$'s become smaller, the resulting $\hat{D}_{n,l}$'s from (\ref{eqn:feasible}) achieve higher SINRs for all the users, but the fronthaul constraints for the RRHs become more difficult to be satisfied. Hence, we propose to apply a simple bisection method to find the optimal values of $\alpha_n$'s, and then substitute $\alpha_n$'s into (\ref{eqn:feasible}) to obtain $\hat{D}_{n,l}$'s. The algorithm is summarized in Table \ref{table4}.

\begin{table}[htp]
\begin{center}
\caption{\textbf{Algorithm \ref{table4}}: Algorithm to Find Feasible Integer Solution $\hat{D}_{n,l}$'s for problem (\ref{eqn:problem 3})} \vspace{0.2cm}
 \hrule
\vspace{0.2cm}

\begin{itemize}
\item[1.] Initialize $\alpha_n^{{\rm min}}=0$, $\alpha_n^{{\rm max}}=1$, $\forall n$;
\item[2.] For $n=1:N$, repeat
\begin{itemize}
\item[a.] Set $\alpha_n=\frac{\alpha_n^{{\rm min}}+\alpha_n^{{\rm max}}}{2}$;
\item[b.] Substitute $\alpha_n$ into (\ref{eqn:feasible}). If $\hat{D}_{n,l}$'s, $\forall l$, satisfy $2B\sum_{l=1}^L\hat{D}_{n,l}\leq \bar{T}_n$, set $\alpha_n^{{\rm max}}=\alpha_n$; otherwise, set $\alpha_n^{{\rm min}}=\alpha_n$;
\end{itemize}
\item[3.] Until $\alpha_n^{{\rm max}}-\alpha_n^{{\rm min}}<\varepsilon$, $\forall n$, where $\varepsilon>0$ is a small constraint to control the accuracy of the algorithm.
\end{itemize}

\vspace{0.2cm} \hrule \label{table4}
\end{center}
\end{table}

{\it 3. Overall Algorithm for Problem (\ref{eqn:problem 11})}

After problems (\ref{eqn:problem 2}) and (\ref{eqn:problem 3}) are solved by Algorithm \ref{table1} and Algorithms \ref{table2} and \ref{table4}, respectively, we are ready to propose the overall algorithm for problem (\ref{eqn:problem 11}) with given $\mv{V}_n=\bar{\mv{U}}_n$, $\forall n$, based on the alternating optimization technique, which is summarized in Table \ref{table3}. The algorithm starts with equal quantization bits allocation for all the spatial filter output dimensions at each RRH, i.e., $D_{n,l}=\lfloor \bar{T}_n/(2BL) \rfloor$, $\forall n,l$. Note that the proposed algorithm terminates in either of the following two cases: 1. $0\leq \gamma^{(i)}-\gamma^{(i-1)}\leq \varepsilon$, i.e., the minimum SINR of users cannot be improved above the positive threshold; and 2. $\gamma^{(i)}-\gamma^{(i-1)}<0$, i.e., the minimum SINR obtained at the current iteration is even reduced compared with that at the previous iteration. The second case may occur since problem (\ref{eqn:problem 3}) is generally not globally solved by Algorithms \ref{table2} and \ref{table4} due to the integer constraints in (\ref{eqn:integer 3}).

\begin{table}[htp]
\begin{center}
\caption{\textbf{Algorithm \ref{table3}}: Algorithm for Problem (\ref{eqn:problem 11}) with given $\mv{V}_n=\bar{\mv{U}}_n$} \vspace{0.2cm}
 \hrule
\vspace{0.2cm}
\begin{itemize}
\item[1.] Initialize: Set $D_{n,l}^{(0)}=\lfloor \bar{T}_n/(2BL) \rfloor$, $\forall n,l$, $\gamma^{(0)}=0$, and $i=0$;
\item[2.] Repeat
\begin{itemize}
\item[a.] $i=i+1$;
\item[b.] Obtain $\mv{p}^{(i)}$ and $\{\mv{w}_k^{(i)}\}$ by solving problem (\ref{eqn:problem 2}) with $\bar{D}_{n,l}=D_{n,l}^{(i-1)}$, $\forall n,l$, using Algorithm \ref{table1};
\item[c.] Obtain $\{\mv{D}_n^{(i)}\}$ by solving problem (\ref{eqn:problem 3}) with $\bar{\mv{p}}=\mv{p}^{(i)}$ and $\bar{\mv{w}}_k=\mv{w}_k^{(i)}$, $\forall k$, using Algorithms \ref{table2} and \ref{table4};
\end{itemize}
\item[3.] Until $\gamma^{(i)}-\gamma^{(i-1)}\leq \varepsilon$, where $\gamma^{(i)}$ denotes the objective value of problem (\ref{eqn:problem 11}) achieved by $\mv{p}^{(i)}$, $\{\mv{w}_k^{(i)}\}$ and $\{\mv{D}_n^{(i)}\}$, and $\varepsilon>0$ is a small constraint to control the accuracy of the algorithm.
\end{itemize}
\vspace{0.2cm} \hrule \label{table3}
\end{center}
\end{table}

\section{Numerical Results}\label{sec:Numerical Results}
In this section, we provide numerical results to verify our results. In the following numerical examples, the bandwidth of the wireless link is $B=10$MHz, while the path loss model of the wireless channel is given as $30.6+36.7\log_{10}(d)$ dB, where $d$ (in meter) denotes the distance between the user and the RRH. The transmit power constraint for each user is $\bar{P}_k=23$dBm, $\forall k$. The power spectral density of the background noise at each RRH is assumed to be $-169$dBm/Hz, and the noise figure due to the receiver processing is $7$dB. Furthermore, it is assumed that all the RRHs possess the identical fronthaul capacity, i.e., $\bar{T}_n=T$, $\forall n$.

\subsection{Effectiveness of Distributed Spatial Filtering at RRHs}\label{sec:Performance Gain due to Beamforming at RRHs}
First, we illustrate the effectiveness of the proposed RRHs' distributed spatial filtering via SCF. In this example, there are $N=4$ RRHs and $K=8$ users randomly distributed in a circle area of radius $500$m. Besides the proposed spatial filter design given in (\ref{eqn:beamforming at RRH}) at each RRH, we also consider the following schemes for performance benchmark: 1. Matched filtering, i.e., $\mv{V}_n=\mv{H}_n^H$, $\forall n$; 2. Zero-forcing (ZF) filtering, i.e., $\mv{V}_n=(\mv{H}_n^H\mv{H}_n)^{-1}\mv{H}_n^H$, $\forall n$; and 3. Without spatial filtering, i.e., $\mv{V}_n=\mv{I}$, $\forall n$. Note that for the matched and ZF filtering, like the proposed design, the signal dimension is reduced from $M$ to $K$ at each RRH, since in this example we set $K<M$.

\begin{figure}
\begin{center}
\scalebox{0.6}{\includegraphics*{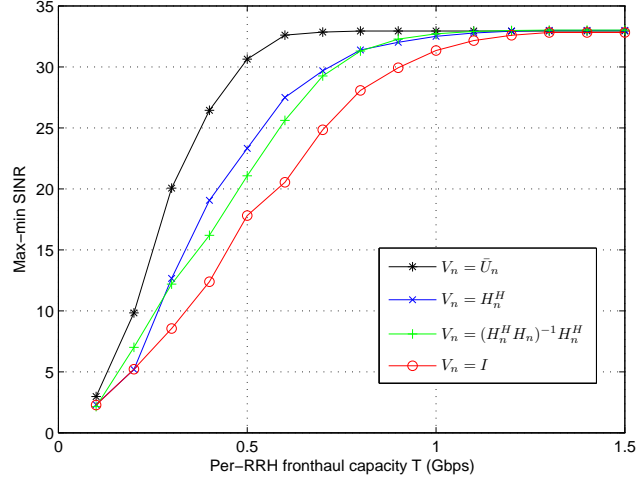}}
\end{center}
\caption{Performance comparison of different distributed spatial filtering designs at RRH versus per-RRH fronthaul capacity.}\label{fig4}
\end{figure}

Fig. \ref{fig4} shows the maximized minimum (max-min) SINR of all the users achieved by Algorithm \ref{table3} with different spatial filtering designs at the RRHs, where each RRH is equipped with $M=10$ antennas and the per-RRH fronthaul capacity $T$ varies from $T=0.1$Gbps to $T=1.5$Gbps. It is observed that the max-min SINR achieved with the proposed spatial filtering in (\ref{eqn:beamforming at RRH}) based on the EVD of the covariance matrix $\mv{S}_n$ (or SVD of $\mv{H}_n$) is higher than that achieved by the other benchmark schemes for all values of fronthaul capacities. This observation also holds in other simulation setups, which are omitted due to the space limitation. The performance gain is explained as follows. For matched filtering, the dimension of the spatial filter output $\bar{\mv{y}}_n$ given in (\ref{eqn:beamforming and quantize}) is reduced from $M$ to $L=K$, but its elements are still correlated at each RRH $n$, which makes the subsequent scalar quantization ineffective. On the other hand, with ZF filtering, although the signal dimension is also reduced and the elements of $\bar{\mv{y}}_n$ are decorrelated in general, the noise effect becomes more severe after filtering and quantization, which degrades the joint decoding SINR at the BBU. Furthermore, with $\mv{V}_n=\mv{I}$, $\forall n$, our proposed SCF scheme reduces to the antenna selection based quantization where the received signals from some antennas are not forwarded to the BBU if the quantization bits allocation obtained from Algorithm \ref{table3} yields $D_{n,m}=0$ for antenna $m$ at RRH $n$. With limited per-RRH fonthaul capacity, this scheme is not effective as substantial user signal information is discarded at each RRH due to antenna selection, even with optimized quantization bits allocation. Last, it is observed that as the per-RRH fronthaul capacity increases to $1.5$Gbps, the max-min SINRs achieved with all spatial filter solutions converge to the same value. This is because when the fronthaul capacity is large enough, the quantization error becomes negligible as compared to the receiver noise at each RRH.

\begin{figure}
\begin{center}
\scalebox{0.6}{\includegraphics*{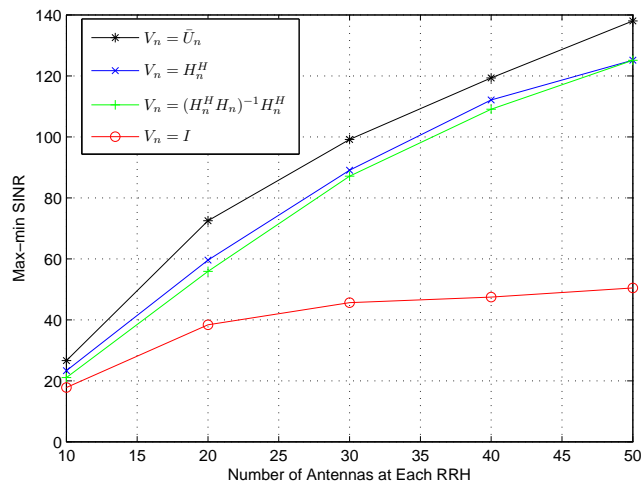}}
\end{center}
\caption{Performance comparison of different spatial filter design at RRH versus per-RRH number of antennas.}\label{fig9}
\end{figure}

Fig. \ref{fig9} shows the max-min SINR of all the users achieved by Algorithm \ref{table3} with different spatial filtering designs at the RRHs versus the number of antennas at each RRH, where the per-RRH fronthaul capacity is fixed as $T=0.5$Gbps. It is observed that the proposed filter outperforms both matched filter and ZF filter for all values of $M$. Moreover, it is also observed that as $M$ increases, the performance gap between the schemes with versus without filtering is enlarged. This is because when $M$ is large, we need to compress the signals received at each RRH more significantly by spatial filtering before scalar quantization, i.e., with $L\ll M$.

\subsection{Performance Gain of Joint Wireless and Fronthaul Resource Allocation}\label{sec:Performance Gain due to Joint Wireless and Fronthaul Resource Allocation}
Next, we show the performance gain of joint optimization for wireless and fronthaul resource allocation. In addition to the proposed Algorithm \ref{table3}, we consider the following three benchmark schemes for performance comparison.

\begin{itemize}
\item{\bf Benchmark Scheme 1: Optimized quantization bits allocation without power control.} In this scheme, all the users transmit at their maximum power, i.e., $p_k=\bar{P}_k$, $\forall k$, while the BBU iteratively updates its MMSE based receive beamforming as given in (\ref{eqn:MMSE}) and RRHs' quantization bits allocation via Algorithms \ref{table2} and \ref{table4}.

\item{\bf Benchmark Scheme 2: Optimized power control with equal quantization bits allocation.} In this scheme, each RRH equally allocates its fronthaul capacity to all the $L$ dimensions of the output signal after spatial filtering, i.e., $D_{n,l}=\bar{T}_n/(2BL)$, $\forall l,n$. Then, the BBU computes the users' power allocation $\mv{p}$ as well as its beamforming solution $\mv{w}_k$'s using Algorithm \ref{table1}.

\item{\bf Benchmark Scheme 3: Equal quantization bits allocation without power control.} In this scheme, all the users transmit at their maximum power, and each RRH equally allocates its fronthaul capacity to all the $L$ dimensions of the filter output signal, while the BBU computes the resulting MMSE receive beamforming.
\end{itemize}

\begin{figure}
\begin{center}
\scalebox{0.6}{\includegraphics*{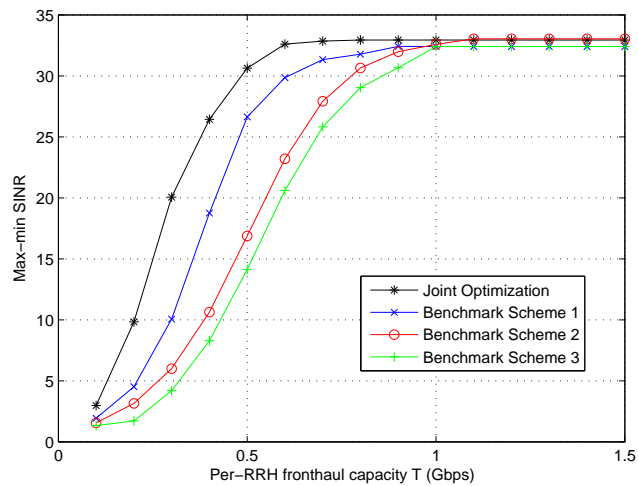}}
\end{center}
\caption{Performance comparison of different resource allocation schemes in the case of small number of users.}\label{fig5}
\end{figure}

Fig. \ref{fig5} shows the performance comparison of different schemes versus per-RRH fronthaul capacity $T$ under the same setup as for Fig. \ref{fig4}. It is observed that the proposed joint optimization of wireless and fronthaul resource allocation, i.e., Algorithm \ref{table3}, achieves higher max-min SINR over Benchmark Schemes 1--3 with separate resource allocation. Moreover, it is observed that under small and moderate per-RRH fronthaul capacity, Benchmark Scheme 1 performs closer to Algorithm \ref{table3}, as compared to Benchmark Schemes 2 and 3, which equally quantize all the signal dimensions at each RRH. In other words, in this example, most of the joint resource allocation gain is due to joint quantization bits allocation at RRHs. Note that in this example, the total number of antennas from all RRHs is $MN=40$, which is much larger than that of the users $K=8$. Since there are sufficient spatial dimensions to separate the users' signals, power control becomes less effective than quantization bits allocation. Table \ref{table5} shows the number of quantized dimensions with $D_{n,l}>0$, $l=1,\cdots,L$, at each RRH $n$ when the per-RRH fronthaul capacity is $T=0.1,0.5,1$Gbps, respectively. It is observed that when $T=0.1$Gbps, each RRH should quantize $3$ out of $L=\min(M,K)=8$ signal dimensions. As the fronthaul capacity increases, more signal dimensions can be quantized. Particularly, when the fronthaul capacity is sufficiently large, i.e., $T=1$Gbps, all the $8$ signal dimensions after spatial filtering should be quantized at all RRHs.

\begin{table}
\caption{Number of Quantized Dimensions at Each RRH} \label{table5}
\begin{center}
\begin{tabular}{c|c|c|c|c}
\hline RRH Index & 1 & 2 & 3 & 4  \\
\hline
No. of Quantized Dimensions (T=0.1Gbps) & 3 & 3 & 3 & 3 \\
\hline
No. of Quantized Dimensions (T=0.5Gbps) & 6 & 7 & 5 & 7 \\
\hline
No. of Quantized Dimensions (T=1.0Gbps) & 8 & 8 & 8 & 8 \\
\hline
\end{tabular}
\end{center}
\end{table}

Fig. \ref{fig6} shows the performance comparison for the case of large number of users with $K=20$. It is observed that Algorithm \ref{table3} still outperforms Benchmark Schemes 1--3. However, different from the case with smaller number of users ($K=8$) shown in Fig. \ref{fig5}, it is observed that with $K=20$, Benchmark Scheme 2 performs better than Benchmark Schemes 1 and 3. This is because with more users but fixed total number of spatial degrees of freedom $MN=40$, the inter-user interference is more severe and thus power control becomes more effective.

\begin{figure}
\begin{center}
\scalebox{0.6}{\includegraphics*{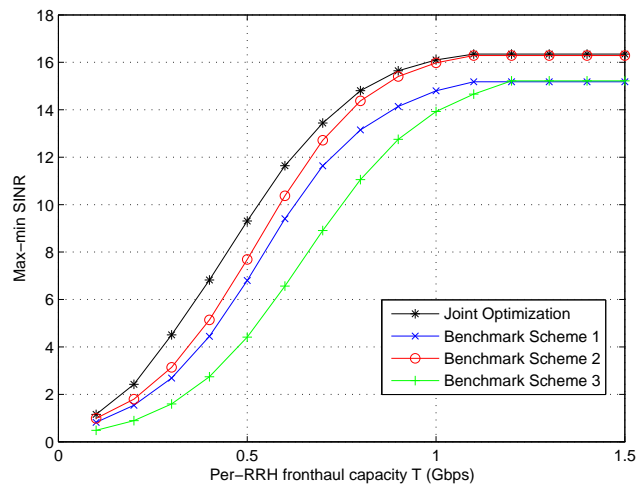}}
\end{center}
\caption{Performance comparison of different resource allocation schemes in the case of large number of users.}\label{fig6}
\end{figure}

\subsection{Multi-Antenna C-RAN versus Massive MIMO}\label{sec:C-RAN or Massive MIMO}
As discussed in Section \ref{sec:Introduction}, besides multi-antenna C-RAN considered in this paper, massive MIMO and single-antenna C-RAN are two promising techniques proposed for 5G wireless networks. Therefore, an interesting as well as important question we seek to address in this subsection is as follows: \emph {given a total amount of antennas to be deployed, should we equip them in one single BS, i.e., massive MIMO, or distribute them over a given area by connecting to the BBU via finite-rate fronthaul links, i.e., C-RAN? Moreover, if C-RAN is preferred, what is the optimal antenna deployment solution, i.e., single-antenna C-RAN versus multi-antenna C-RAN?} Intuitively, if more single-antenna RRHs are deployed in the network, with higher probability each user can be served by one or more nearby RRHs with strong channel conditions. However, with multi-antenna RRHs, we can efficiently perform SCF at each RRH to better utilize the limited fronthaul capacity given for each RRH. In the following, we provide a case of study to show the advantage of multi-antenna C-RAN over its two extreme counterparts: massive MIMO with all the antennas deployed at one BS and single-antenna C-RAN with only one antenna at each RRH, i.e., $M=1$.

To make a fair comparison, we assume that there are in total $\bar{M}$ antennas to serve $K$ users in a given area. Specifically, for the massive MIMO system, we assume that there is only one BS which is equipped with all $\bar{M}$ antennas, while for the C-RAN, we assume that there are $N$, $1\leq N \leq \bar{M}$, RRHs each equipped with $\bar{M}/N$ antennas.\footnote{In the case that $\bar{M}/N$ is not an integer, we can simply let some RRHs possess $\lfloor \bar{M}/N \rfloor$ antennas, while the others possess $\lceil \bar{M}/N \rceil$ antennas such that the total number of antennas is equal to $\bar{M}$.} Moreover, we assume there is a total fronthaul capacity for all RRHs in C-RAN, denoted by $\bar{T}$, which is a constant regardless of the number of RRHs $N$. Note that this is a valid assumption since in practice the signals from different RRHs in the same area are usually first multiplexed at a local hub, which then forwards the signals to the BBU via a fronthaul link with capacity $\bar{T}$. For simplicity, we set $\bar{T}_n=\bar{T}/N$ for all the $N$ RRHs in the case of C-RAN.

\begin{figure}
\begin{center}
\scalebox{0.6}{\includegraphics*{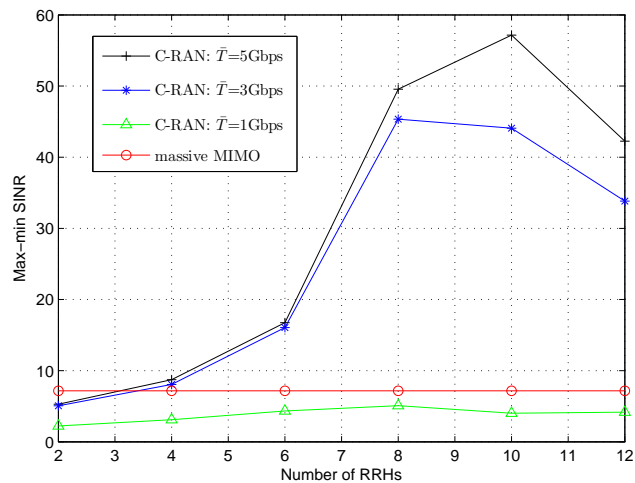}}
\end{center}
\caption{Performance comparison between multi-antenna C-RAN and massive MIMO.}\label{fig7}
\end{figure}

Fig. \ref{fig7} shows the max-min SINR performance comparison between massive MIMO versus C-RAN with different number of RRHs and fronthaul sum-capacity. The total number of antennas is $\bar{M}=50$. Moreover, there are $K=20$ users randomly located in a circle area of radius $700$m. For the massive MIMO, the BS is located in the center of the circle, while for the C-RAN, the RRHs are randomly located in the circle. Note that since the BS in massive MIMO system directly decodes the user messages, the optimal power allocation for users and decoding beamforming vectors at the BS can be obtained by Algorithm \ref{table1} with $\bar{D}_{n,l}=\infty$, $\forall n,l$. It is observed from Fig. \ref{fig7} that if the fronthaul sum-capacity is sufficiently large, i.e., $\bar{T}=3$Gbps or $5$Gbps, the max-min SINR in C-RAN is larger than that in massive MIMO when the number of RRHs $N> 4$. This is because the densification gain due to larger number of RRHs dominates the performance of C-RAN. However, if the fronthaul sum-capacity is limited, i.e., $\bar{T}=1$Gbps, the performance of C-RAN with any $N$ is inferior to that of massive MIMO, since quantization errors dominate its performance. Moreover, for the case of single-antenna C-RAN with $N=50$ single-antenna RRHs (which is not shown in Fig. \ref{fig7} for brevity), the achieved max-min SINRs for the cases of $\bar{T}=1,3,5$Gbps are $2.25$, $10.23$ and $22.08$, respectively, which are significantly lower than the maximum SINRs achieved in multi-antenna C-RAN with the corresponding optimal number of RRHs shown in Fig. \ref{fig7}. This is because in the case of single-antenna C-RAN, the received signals at nearby RRHs are highly correlated, and as a result independent scalar quantization at these RRHs without spatial filtering is inefficient in utilization of the given fronthaul sum-capacity. Last, it is observed that the achieved max-min SINR in multi-antenna C-RAN first increases with the number of RRHs, but then decreases after a certain optimal number of RRHs for all three cases of $\bar{T}=1,3,5$Gbps. This validates the effectiveness of multi-antenna C-RAN by more flexibly balancing between the performance and fronthaul trade-off. As $\bar{T}$ increases, it is observed that more RRHs should be deployed in C-RAN to exploit better channels from all users.

\section{Conclusions}\label{sec:Conclusions}
This paper considers a flexible antenna deployment design for C-RAN termed multi-antenna C-RAN and proposes a new ``spatial-compression-and-forward (SCF)'' scheme for efficient and low-complexity processing at each RRH in the uplink multiuser communication. With the proposed distributed spatial filters at RRHs, a joint optimization across the wireless transmission, the fronthaul quantization and the decoding at the BBU is performed to maximize the minimum SINR of all the users. Our results show that the proposed SCF scheme with joint resource allocation achieves significant performance gains over the conventional ``quantize-and-forward'' based single-antenna C-RAN as well as massive MIMO, thanks to the flexibly optimized antenna distribution in multi-antenna C-RAN.

\begin{appendix}

\subsection{Proof of Lemma \ref{lemma1}}\label{appendix1}
First, since $\left(\sum\limits_{j\neq k}p_j\tilde{\mv{h}}_{(j)}\tilde{\mv{h}}_{(j)}^H+\sigma^2\mv{I}+\mv{Q}(\mv{p},\{\bar{\mv{U}}_n,\bar{\mv{D}}_n\})\right)^{-1}\succ \mv{0}$, $\forall k$, we have $\mv{I}(\mv{p})>\mv{0}$ if $\mv{p}>\mv{0}$. The first property in Lemma \ref{lemma1} is proved.

Next, if $\mv{p}\geq \mv{p}'$, then we have
\begin{align}
& \sum\limits_{j\neq k}p_j\tilde{\mv{h}}_{(j)}\tilde{\mv{h}}_{(j)}^H+\sigma^2\mv{I}+\mv{Q}(\mv{p},\{\bar{\mv{U}}_n,\bar{\mv{D}}_n\})\nonumber \\ \succeq & \sum\limits_{j\neq k}p_j'\tilde{\mv{h}}_{(j)}\tilde{\mv{h}}_{(j)}^H+\sigma^2\mv{I}+\mv{Q}(\mv{p}',\{\bar{\mv{U}}_n,\bar{\mv{D}}_n\}), ~~~ \forall k.
\end{align}

\begin{lemma}\label{lemma2}{\cite[Corollary 7.7.4]{Horn85}}
If $\mv{A}\succ \mv{0}$, $\mv{B}\succ \mv{0}$, and $\mv{A}-\mv{B} \succeq \mv{0}$, then $\mv{A}^{-1}-\mv{B}^{-1}\preceq \mv{0}$.
\end{lemma}

According to Lemma \ref{lemma2}, it follows that
\begin{align}
& \left(\sum\limits_{j\neq k}p_j\tilde{\mv{h}}_{(j)}\tilde{\mv{h}}_{(j)}^H+\sigma^2\mv{I}+\mv{Q}(\mv{p},\{\bar{\mv{U}}_n,\bar{\mv{D}}_n\})\right)^{-1}\nonumber \\ \preceq &  \left(\sum\limits_{j\neq k}p_j'\tilde{\mv{h}}_{(j)}\tilde{\mv{h}}_{(j)}^H+\sigma^2\mv{I}+\mv{Q}(\mv{p}',\{\bar{\mv{U}}_n,\bar{\mv{D}}_n\})\right)^{-1}, ~~~ \forall k.
\end{align}As a result, we have $\mv{I}(\mv{p})\geq \mv{I}(\mv{p}')$ if $\mv{p}\geq \mv{p}'$. The second property in Lemma \ref{lemma1} is thus proved.

Last, if $a>1$, then we have
\begin{align}
& a\left(\sum\limits_{j\neq k}p_j\tilde{\mv{h}}_{(j)}\tilde{\mv{h}}_{(j)}^H+\sigma^2\mv{I}+\mv{Q}(\mv{p},\{\bar{\mv{U}}_n,\bar{\mv{D}}_n\})\right)\nonumber \\ \succ & \sum\limits_{j\neq k}ap_j\tilde{\mv{h}}_{(j)}\tilde{\mv{h}}_{(j)}^H+\sigma^2\mv{I}+\mv{Q}(a\mv{p},\{\bar{\mv{U}}_n,\bar{\mv{D}}_n\}), ~~~ \forall k.
\end{align}According to Lemma \ref{lemma1}, we have
\begin{align}
& \frac{1}{a}\left(\sum\limits_{j\neq k}p_j\tilde{\mv{h}}_{(j)}\tilde{\mv{h}}_{(j)}^H+\sigma^2\mv{I}+\mv{Q}(\mv{p},\{\bar{\mv{U}}_n,\bar{\mv{D}}_n\})\right)^{-1}\nonumber \\ \prec & \left(\sum\limits_{j\neq k}ap_j\tilde{\mv{h}}_{(j)}\tilde{\mv{h}}_{(j)}^H+\sigma^2\mv{I}+\mv{Q}(a\mv{p},\{\bar{\mv{U}}_n,\bar{\mv{D}}_n\})\right)^{-1}, ~~~ \forall k.
\end{align}As a result, it follows that $\forall a>1$, $a\mv{I}(\mv{p})>\mv{I}(a\mv{p})$. The third property of Lemma \ref{lemma1} is thus proved.

To summarize, $\mv{I}(\mv{p})$ given in (\ref{eqn:interference function}) is a standard interference function. Lemma \ref{lemma1} is thus proved.

\subsection{Proof of Corollary \ref{corollary3}}\label{appendix5}
If the initial point of the fixed-point method is $\mv{p}^{(0)}=\mv{0}$, then we have $\mv{p}^{(1)}=\mv{I}(\mv{p}^{(0)})>\mv{0}=\mv{p}^{(0)}$. Moreover, according to the second property of Lemma \ref{lemma1}, if $\mv{p}^{(i)}\geq\mv{p}^{(i-1)}$, then $\mv{p}^{(i+1)}=\mv{I}(\mv{p}^{(i)})\geq\mv{I}(\mv{p}^{(i-1)})=\mv{p}^{(i)}$. As a result, with the initial point $\mv{p}^{(0)}=\mv{0}$, the resulted power solution increases after each iteration, i.e., $\mv{p}^{(0)}<\mv{p}^{(1)}\leq \mv{p}^{(2)}\leq \cdots$. Suppose that at the $i$th iteration, the above procedure converges to a finite power solution $\mv{p}^{(i)}$. Then we have $\mv{p}^{(i+1)}=\mv{p}^{(i)}=\mv{I}(\mv{p}^{(i)})$. As a result, $\mv{p}^{(i)}$ can satisfy the SINR constraints (\ref{eqn:SINR balancing 7}), which contradicts to the fact that $\bar{\gamma}$ cannot be achieved by all the users. To summarize, the power solution by the fixed-point method with $\mv{p}^{(0)}=\mv{0}$ will increase to an infinity power allocation. Corollary \ref{corollary3} is thus proved.

\subsection{Proof of Corollary \ref{corollary2}}\label{appendix2}
Let $\mv{p}^{(0)}$ denote a feasible solution to problem (\ref{eqn:problem 7}). Since it satisfies (\ref{eqn:SINR balancing 7}), we have $\mv{p}^{(0)}\geq \mv{I}(\mv{p}^{(0)})=\mv{p}^{(1)}$. In the following, we show that $\mv{p}^{(i+1)}\leq \mv{p}^{(i)}$ holds $\forall i\geq 1$ by induction. Suppose $\mv{p}^{(i)}\leq \mv{p}^{(i-1)}$. Then we have $\mv{p}^{(i+1)}=\mv{I}(\mv{p}^{(i)})\overset{(a)}{\leq}\mv{I}(\mv{p}^{(i-1)})=\mv{p}^{(i)}$, where $(a)$ is due to the second property in Lemma \ref{lemma1}. Since we have $\mv{p}^{(1)}\leq \mv{p}^{(0)}$, we thus have $\mv{p}^{(i+1)}\leq \mv{p}^{(i)}$ holds $\forall i\geq 1$. In other words, given any feasible power solution $\mv{p}'$ as the initial point, a monotonic convergence can be guaranteed by the fixed-point method, i.e., $\mv{p}'\geq \mv{p}^\ast$. Corollary \ref{corollary2} is thus proved.

\end{appendix}

\end{document}